\documentclass[12pt,a4page]{article}
\usepackage[T1]{fontenc}
\usepackage{lmodern, microtype}
\usepackage[left=1.25in, right=1.25in, top=1.25in, bottom=1.25in]{geometry}
\usepackage[small]{titlesec}
\usepackage{mathtools,comment,paralist}
\usepackage{epigraph}

\usepackage{amssymb,amsmath,amsthm,fancyhdr,bbm}
\usepackage{dsfont} %for indicator
\usepackage{float}
\usepackage{subcaption}
\usepackage{xcolor}
\usepackage{hyperref}
\hypersetup{
  colorlinks=true,
  linkcolor=blue,
  citecolor=black
}
\usepackage[capitalise]{cleveref}
\usepackage{natbib}
\setcitestyle{authoryear,open={(},close={)}}
\makeatletter

\newtheorem{theorem}{Theorem}
\newtheorem{lemma}{Lemma}
\newtheorem*{lemma*}{Lemma}

\newtheorem*{axiom*}{Axiom}
\newtheorem{proposition}{Proposition}

\newtheorem{corollary}{Corollary}
\newtheorem{example}[theorem]{Example}

\newtheorem{maintheorem}{Theorem}
\newtheorem*{theorem*}{Theorem}

\newtheorem{definition}{Definition}
\newtheorem*{definition*}{Definition}
\usepackage{graphicx}
\usepackage{booktabs}
\usepackage{pstricks, enumerate, pst-node, pst-text, pst-plot}

\usepackage{mleftright}
\usepackage{xparse}
\DeclareDocumentCommand\Pr{ m g }{\ensuremath{
    {   \IfNoValueTF {#2}
      {\mathbb{P}\mleft[{#1}\mright]}
      {\mathbb{P}\mleft[{#1}\middle\vert{#2}\mright]}
    }
}}
\DeclareDocumentCommand\E{ m g }{\ensuremath{
    {   \IfNoValueTF {#2}
      {\mathbb{E}\mleft[{#1}\mright]}
      {\mathbb{E}\mleft[{#1}\middle\vert{#2}\mright]}
    }
}}

\newcommand{\R}{\mathbb{R}}

\newcommand{\mP}{\mathbb{P}}

\newcommand{\abs}{\ge_{\rm abs}}
\newcommand{\rel}{\ge_{\rm rel}}
\newcommand{\norel}{\ngeq_{\rm rel}}

\newcommand{\ccv}{\ge_{\rm c}}

\newcommand{\sccv}{>_{\rm c}}

\usepackage{accents}

\def\ee{\mathrm{e}}
\def\cara{\mathrm{CARA}}
\def\crra{\mathrm{CRRA}}

\def\scara{\mathrm{cara}}
\def\scrra{\mathrm{crra}}
\def\shara{\mathrm{hara}}

\def\cA{\mathcal{A}}

\def\cU{\mathcal{U}}
\def\cL{\mathcal{L}}

\def\cF{\mathcal{F}}

\def\cU{\mathcal{U}}

\def\sm{C^1_+}

\ifdefined\DRAFT

\definecolor{ForestGreen}{rgb}{.13,.54,.13}
\definecolor{violet}{cmyk}{0.79,0.88,0,0}
\definecolor{darkmagenta}{rgb}{0.55, 0.0, 0.55}

\definecolor{darkBrickRed}{rgb}{.70,.13,.16}
\newcommand{\po}[1]{{\color{red}{(\textbf{Po:} #1)}}}
\newcommand{\ben}[1]{{\color{red}{(\textbf{Ben:} #1)}}}

\else
\newcommand{\fed}[1]{}

\newcommand{\omer}[1]{}
\newcommand{\benpo}[1]{}
\newcommand{\ben}[1]{}
\newcommand{\po}[1]{}
\fi

\definecolor{redPart}{rgb}{1.0, 0.0, 0.0}
\definecolor{greenPart}{rgb}{0.0, 1, 0.0}
\definecolor{bluePart}{rgb}{0.0, 0.0, 1.0}
\definecolor{yellowPart}{rgb}{1, 1, 0}
\definecolor{magentaPart}{rgb}{1, 0, 1}
\definecolor{cyanPart}{rgb}{0.0, 1, 1}

\DeclareMathOperator*{\argmax}{arg\!max}

\title{Risk Aversion Reversals\footnote{
We are grateful to Omer Tamuz, Charles Sprenger, Luciano Pomatto, Kirby Nielsen, and Ted O'Donoghue for their guidance and support, as well as Jose Apesteguia, Miguel Ballester, and Tomasz Strzalecki for their valuable comments. We also thank Marina Agranov, Aislinn Bohren, Peter Caradonna, Federico Echenique, Michael Gibilisco, Philip Haile, Jay Lu, Aldo Lucia, Axel Niemeyer, Tom Palfrey, Kota Saito, Christopher Turansick, Jeffrey Yang, Leeat Yariv, and the audiences at the ESA conference, LAX workshop, LA Theory conference, and DTEA conference for their helpful comments and suggestions. 
}}

\author{Po Hyun Sung\thanks{Caltech. Email: psung@caltech.edu}\\
Ben Wincelberg\thanks{Caltech. Email: bwincelb@caltech.edu}}

\date{\today\\
}
\begin{document}
\maketitle
\begin{abstract}
Standard stochastic choice models used to estimate risk aversion can lead to risk-aversion reversals, where a more risk-averse individual chooses a riskier lottery more frequently than a less risk-averse individual. We study when reversals are implied by the preference specification rather than the noise specification. We say that two utilities imply reversals in a given noise framework if reversals arise for every specification of noise for each individual. For weak utility, a flexible class that includes logit and probit and allows for menu-dependent noise, two utilities imply reversals if and only if their curvature ratio is unbounded. This condition holds for CARA, CRRA, and their generalizations, for which reversals arise for empirically relevant coefficients and lotteries, raising concerns about resulting estimates and out-of-sample predictions. Finally, we show that equicautious HARA, sum-ex, and sum-power utilities do not imply reversals and that, moreover, these families remain well-behaved for multinomial choice.

\end{abstract}
\section{Introduction}
A precise understanding of how humans evaluate risk is central to economics. 
Since choices are inherently noisy, analysts deploy stochastic choice models to recover individuals' risk preferences and noise from choice data. The most widely used stochastic choice models for this purpose are \textit{Fechnerian models}, a non-parametric class that includes the logit and probit models. In practice, the analyst chooses a functional form for utility and a functional form for noise which maps utility differences to choice probabilities \cite*[e.g.,][]{harrison2007naturally,von2011heterogeneity,holzmeister2021risk}.

However, predictions from these models under CARA or CRRA---the most commonly used utility functions---suffer from a well-known problem. When two individuals differ in their risk-aversion parameters but face the same noise, in many choices reversals occur: the more risk-averse individual exhibits more risk-seeking behavior \citep{wilcox2008stochastic,wilcox2011stochastically,blavatskyy2011probabilistic,apesteguia2018monotone}.  In response to these critiques, subsequent work has argued that non-monotonicity is particularly worrisome under homoskedastic noise and has advocated allowing heteroskedasticity to better calibrate noise to utility \citep*{barseghyan2018estimating,keffert2024stochastic}.  Moreover, it is common to model heterogeneity in both preferences and noise and to estimate them jointly \cite*[e.g.,][]{hey1994investigating,von2011heterogeneity,meissner2023individual}. 

In this paper, we
show that the problems noted for Fechnerian models are far more pervasive than previously recognized: they persist even when individuals differ arbitrarily in their Fechnerian noise, when noise varies across menus, and when noise takes a more general, non-Fechnerian form. 
That is, under CARA and CRRA utilities, none of these generalizations can eliminate the counterintuitive predictions. We then characterize which utility functions, beyond CARA and CRRA, give rise to risk-aversion reversals and advocate for parametric utility families that do not, under various noise specifications, exhibit reversals: a more risk-averse individual always chooses safer options more frequently. Beyond its intuitive appeal, this comparative static aligns with empirical findings \citep{bruner2017does}. Finally, we discuss a generalization of Fechnerian models to multinomial choice and show that our proposed utility families do not exhibit reversals even in this more general setting.

\medskip

To build intuition, we begin with the baseline Fechnerian framework, in which each individual is associated with a Bernoulli utility function $u$ and a strictly increasing function $F$. We refer to $F$ as a Fechnerian noise structure, where
\[
F(U(X)-U(Y))
\] is the probability that $X$ is chosen over $Y$, and $U(X)=\E{u(X)}$. The 
most popular Fechnerian noise structures are CDFs of normal and logistic distributions, corresponding to probit and logit models. Indeed, any additive random utility model with i.i.d.\ shocks fits the Fechnerian framework.\footnote{Technically, for $F$ to be strictly increasing, the difference in these shocks must have full support.}

Our first result is that, for some pairwise choices between a safe and a risky lottery,\footnote{We say that lottery $X$ is safer than lottery $Y$ if $X$ dominates $Y$ in the concave order. Equivalently, every expected utility maximizer with a concave utility function must prefer $X$ to $Y$.}  a more risk-averse individual chooses the risky option more often than a less risk-averse individual, even when the two differ arbitrarily in their Fechnerian noise structures (\Cref{th:main}). Thus, the problems identified for the homoskedastic noise model cannot be resolved by jointly estimating risk and noise, even when noise is estimated non-parametrically. 

For example, consider an analyst who observes Anne and Bob choosing from three distinct menus---each menu consists of a safe lottery---a sure payoff of \$8, \$10, or \$14---and a 50–50 risky lottery paying plus or minus \$4, \$6, or \$8 from the sure amount. The choice probabilities for the safe lotteries are shown in the second and third columns of \Cref{tab:crra_reversal_scale}. Assuming CRRA expected utilities for their risk preferences and normal distributions for their noise structures, Anne is estimated to be more risk-averse and experience less noise than Bob. Indeed, the choice probabilities in \Cref{tab:crra_reversal_scale} correspond to Anne having a CRRA coefficient of $0.8$ and noise variance of $0.5$, compared to Bob’s $0.3$ and $1$, respectively.
\footnote{I.e., the analyst deploys the probit model for Anne and Bob \begin{equation*}
F_{A}(U_A(X)-U_A(Y)) \quad \text{and}\quad  F_{B}(U_B(X)-U_B(Y)),
\end{equation*} 
 respectively, where $U_A$ and $U_B$ are CRRA expected utility functions with relative risk aversions $a$ and $b$, and  $F_{A}$ and $F_{B}$ are the CDFs of normal distributions with zero means and respective variances $\sigma_{A}^2$ and $\sigma_{B}^2$.}

        \begin{table}[H]
\centering
\begin{tabular*}{\textwidth}{@{\extracolsep{\fill}} l cccccc}
\toprule
$S$ vs $R$ &$\rho_{A}(S,R)$ & & $\rho_{B}(S,R)$ & $\rho_{A}(3S,3R)$ &  &$\rho_{B}(3S,3R)$\\
\midrule
$8$ vs $4, 12$ & 0.60 & > & 0.57 & 0.62 & < & 0.65 \\
$10$ vs $4, 16$ & 0.65 & > & 0.62  & 0.69 & < & 0.74\\
$14$ vs $6, 22$ & 0.64 & > & 0.63  & 0.68 & < & 0.77 \\
\bottomrule
\end{tabular*}
\caption{CRRA Choice Probabilities for $S$ and $3S$ with normal errors. }
\label{tab:crra_reversal_scale}
\end{table}

 In the last two columns of \Cref{tab:crra_reversal_scale}, we report the choice probabilities of safe options whose outcomes are tripled from the base lotteries under the estimated parameters. Note that when the stakes are increased, Bob chooses safer options more frequently than Anne from all menus. 
 Thus, despite Anne being estimated as more risk-averse and more precise, the model predicts that she will behave more risk-seeking than Bob once the stakes are scaled up. 
 
This reversal pattern does not hinge on these specific CRRA risk coefficients and noise structures, nor on the specific choice of lotteries. Indeed, given lotteries $X$ and $Y$ with $X$ safer than $Y$, we show how to construct lotteries $X'$ and $Y'$ offering higher potential rewards, such that $X'$ remains safer than $Y'$ yet leads to a reversal. 
For CARA utilities, reversals occur after shifting the lotteries by the same sufficiently large outcome (\Cref{prop:stakes}). Thus, CARA and CRRA Fechnerian models with arbitrary noise structures make the perverse and systematic prediction that those most inclined toward choosing safer options when stakes are small must become the least inclined when stakes are larger.

Assuming that this reversal pattern does not always bear out in reality---and indeed  \cite{bruner2017does} finds that it does not---our results indicate that the estimates yielded by these models will be sensitive to the level of stakes that are analyzed, making them unreliable.\footnote{This echoes the critique of expected utility in \cite{rabin2000risk}, which also reveals paradoxes when stakes are increased. While Rabin's result is about the curvature of utility functions in deterministic models, our result is about the interaction between noise structures and CARA/CRRA utilities.} Indeed, if the analyst instead observed the choice probabilities of the scaled lotteries in \Cref{tab:crra_reversal_scale}, the analyst would estimate that Bob is more risk-tolerant than Anne, even if he exhibited more risk-averse behavior. 
As illustrated in that table, the stake increases are often modest, indicating that such reversals can emerge even in low-stakes experimental settings. Moreover, these models give counterintuitive, and perhaps counterfactual, out-of-sample predictions for choices made at larger or smaller stakes.

\medskip
We next show that the paradoxical reversals of CARA and CRRA Fechnerian models are not artifacts of overly restrictive Fechnerian noise structures. To establish this, we explore \emph{weak expected utility} (WEU) models. A WEU model consists of a vNM utility function $U$, a Fechnerian noise structure $F$ and a menu-dependent scale parameter $\sigma \colon \cL \times \cL \to (0,\infty)$ such that the probability of choosing $X$ over $Y$ is given by \[F\left(\frac{U(X)-U(Y)}{\sigma(X,Y)}\right).\]

WEU models generalize the standard Fechnerian framework by allowing choice probabilities to depend on factors beyond strict utility differences. An important subclass of WEU models arises in the additive random utility framework, when the random shocks to utility that are associated with each lottery are independent but not identically distributed. Economically, this accommodates context-dependent noise; for example, large-stakes lotteries may induce increased attention, resulting in less variable shocks.

Without any restrictions on how noise varies across menus, the only implication of WEU models is that individuals will choose their preferred lottery more than half the time. Thus, to derive meaningful comparisons between individuals with different risk attitudes, we must impose additional structure. Our next results (\cref{th:finalboss} and \cref{th:menu_dep}) show that, under a mild continuity requirement on the menu-dependent scale parameters, WEU models based on CARA or CRRA utilities are still guaranteed to produce reversals. Thus, even when allowing for arbitrary noise that varies across both individuals and menus, the CARA and CRRA expected utility functional forms lead to untenable predictions.

 One might suspect, at this point, that these reversals arise because $F$ and $\sigma$ are applied to vNM utilities, which are just representations of preferences and carry no cardinal significance. A natural alternative would be to apply $F$ and $\sigma$ to certainty equivalents $u^{-1}(\E{u(X)})$, or more generally to $f(\E{u(X)})$ for some strictly increasing transformation $f$. However, \Cref{cor:transform} establishes that this transformation does not resolve the problem. Ultimately, CARA and CRRA utilities prove to be fundamentally incompatible with stochastic choice consistency.

\medskip

We next turn our attention toward the characterization of families of Bernoulli utilities that interact well with noise. Our main result is that when $u_A$ is more risk-averse than $u_B$, there exist reversals for every continuous weak utility model if and only if the ratio
\[u_B''(x)/u_A''(x)\] is unbounded (\Cref{th:finalboss}). 
Indeed, for CARA and CRRA utilities, this condition holds, leading to reversals in every model of noise we have considered. Moreover, this condition is met by a class of utility functions that generalizes CARA and CRRA (\cref{cor:expo}).

We thus advocate for using parametric utility functions that do not carry the implication of reversals. These include families generated by mixtures of two CARA or CRRA utility functions as well as the class of equicautious HARA utilities (see \cref{sec:mono}). These simple families of utility functions offer well-behaved alternatives to CARA and CRRA for empirical analysis.

We conclude by studying reversals in the more general multinomial-choice setting. 
We focus on the \emph{additive perturbed utility} (APU) model of \cite*{fudenberg2015stochastic}, a multinomial extension of the Fechnerian model. Our main result is that the utility families that avoid reversals in the binary setting also avoid reversals in the multinomial setting (\cref{th:apu_rel}).

\subsection{Related Literature}
We contribute to a large body of literature on risky choices in the presence of noise \citep*[e.g.,][]{becker1963stochastic,harless1994predictive,hey1994investigating,loomes2002microeconometric,blavatskyy2007stochastic}. See \cite{wilcox2021utility} for a recent survey. In particular, we study Fechnerian noise structures and their generalizations. Axiomatic investigations of these models have been undertaken by \cite{debreu1958stochastic},  \cite{tversky1969substitutability}, \cite{he2024moderate}, and \cite*{fudenberg2015stochastic}, among others. See \cite{Strzalecki_2025} for an extensive review.

Our main result concerns WEU models, which recognize that some comparisons may involve greater cognitive effort or uncertainty, allowing noise to depend on the menu. \cite{hey1995experimental}, \cite{buschena2000generalized}, and \cite{loomes2005modelling} study WEU models where noise depends on, for example, value differences between lotteries and question difficulty. More recently, \cite{he2024moderate} and \cite{shubatt2024tradeoffs} characterize noise that depends on a measure of distance between alternatives.

\medskip
One of the first papers to study comparative risk aversion in stochastic choice is \cite{blavatskyy2011probabilistic}, who examines when one individual is always more likely to choose a sure thing over a lottery than another individual. His main result is that this notion of stochastic comparative risk aversion is too strong to be compatible with any Fechnerian expected utility models.

The papers most closely related to ours are \cite{wilcox2011stochastically} and \cite{apesteguia2018monotone}, which demonstrate that CARA and CRRA utilities coupled with heterogeneous Fechnerian noise are problematic since they lead to reversals. \cite{wilcox2011stochastically} shows that other parameterizations of CARA and CRRA utilities also lead to reversals. This corresponds to two individuals having noise structures that are related by a scale factor. We study this special case in \Cref{rem:scale}. 

\cite{apesteguia2018monotone} study a larger order that includes the concave order: they rank one lottery as safer than another if increasing the underlying risk-aversion parameter can only shift preferences in its favor. Thus, this definition depends on the underlying family of utilities, unlike the concave order. \cite{apesteguia2018monotone} establish that when the utility difference for a pair of ordered lotteries is non-monotone in the risk parameter, there is a pair of risk-aversion parameters where the higher parameter chooses the riskier option more frequently than the lower parameter in homoskedastic Fechnerian models. They show that for CARA and CRRA, the utility difference is non-monotone for every pair of ordered lotteries.

Our work builds on the insights of \cite{wilcox2011stochastically} and \cite{apesteguia2018monotone} in several ways. First, we allow for heterogeneous noise structures and show that the paradox of reversals persists without any parametric assumptions on how noise differs across individuals (\Cref{th:main}). Second, we show that the reversals occur under WEU models, in which noise varies across menus (\Cref{th:menu_dep}). 
Finally, we characterize which families of utility functions, beyond CARA and CRRA, suffer from these reversals (\Cref{th:finalboss}) and propose alternative parametric families of utilities that do not suffer from the problem. This is in contrast with the solutions proposed by \cite{wilcox2011stochastically} and \cite{apesteguia2018monotone} who retain CARA and CRRA utilities.

\cite{wilcox2011stochastically} proposes the \emph{contextual utility} model---a special case of WEU---and shows that lotteries over three fixed outcomes do not generate reversals. Nevertheless, \cite{apesteguia2018monotone} show that this model is no longer monotone for more than three outcomes and propose instead the random parameter model, where each individual is associated with a distribution over CARA or CRRA preferences and chooses probabilistically as if they draw a random preference from their distribution. 

As \cite{apesteguia2018monotone} point out, random parameter models trivially do not give rise to reversals in binary choice that involve pairs that are stochastic dominance related, since they predict deterministic choices in this case, as every risk-averse individual prefers the safer option.\footnote{To avoid this, a second stage of randomness is often added in practice through a trembling parameter, see, e.g., \cite{loomes2002microeconometric,apesteguia2018monotone,keffert2024stochastic,jagelka2024economists}.} As these models are not equipped to study individual variation in stochastic choice behavior between safer and riskier alternatives, they are not a focus in our study of risk-aversion reversals.

\medskip

In the deterministic setting,
\cite*{kihlstrom1981risk} and \cite{ross1981some} note that the risk premium\footnote{Risk premium is defined as the amount subtracted from a less risky lottery to establish indifference with a riskier one in the concave order} is not necessarily increasing in the Arrow-Pratt order.  \cite{keffert2024stochastic} take this insight to the stochastic setting and argue that precluding reversals for concave-ordered lotteries is overly restrictive, leading them to propose an alternative notion of reversals. 

Whereas \cite{keffert2024stochastic} deal with the paradox of reversals by proposing a weaker requirement on stochastic choice behavior, an alternative is to restrict the set of underlying preferences. Indeed, to solve the deterministic paradox of risk-premia, \cite{ross1981some} proposes a stronger notion of comparative risk aversion than Arrow-Pratt that successfully restores monotonicity. We bring this preference-based approach into the stochastic domain. As our characterization result (\Cref{th:finalboss}) shows, Ross' stronger notion implies our condition, demonstrating that it is sufficient to eliminate choice reversals under certain noise structures.

We also contribute to a literature examining failures of expected utility theory in predicting choices across varying ranges of stakes. The most prominent critique in this area is by \cite{rabin2000risk}, who demonstrates that within the expected utility framework, even modest levels of risk aversion exhibited over small-stakes, imply absurd, counterfactual levels of risk aversion for large stakes. While this result holds for deterministic choice, \cref{rem:scale,prop:stakes} echo \cite{rabin2000risk} in the stochastic setting: out of sample predictions made from small-stakes observations are counterintuitive.
\section{Preliminaries}

 We denote by $\mathcal L$ the set of all bounded real random variables defined over a non-atomic probability space $(\Omega,\Sigma,\mathbb{P})$. We use the term \textit{lotteries} to refer to elements of $\cL$. A stochastic choice rule is a map $\rho\colon \cL \times \cL \to \R$  where $\rho(X,Y)+\rho(Y,X)=1$ and $\rho(X,Y)$ is the probability of choosing $X$ over $Y$. A vNM \emph{utility function} is a map $U\colon \cL \to \R$ given by $U(X) = \E{u(X)}$ for some strictly increasing, concave, and continuous Bernoulli utility  $u \colon \R \to \R$. 

We consider an individual that evaluates lotteries according to a vNM utility $U$ but faces noise and chooses stochastically according to $\rho$. We think of $\rho$ as a noisy expression of the underlying utility $U$. 
For example, an individual with utility $U$ whose noise takes a logit form has choice probabilities 
\[\rho(X,Y)=\frac{\ee^{\beta U(X)}}{\ee^{\beta U(X)}+\ee^{\beta U(Y)}},\] for some $\beta > 0$. When $\beta$ is held constant across all menus, $\rho$ belongs to the class of Fechnerian models, where choice probabilities are a function of utility differences. We discuss Fechnerian models in depth in \cref{sec:main}. More generally, the logit scale parameter $\beta$ may depend on $(X,Y)$. In this case, there is much more flexibility. Indeed, the only restrictions of this formulation are that choice probabilities are in $(0,1)$ and that $\rho(X,Y)\ge \frac{1}{2}\iff U(X)\ge U(Y)$, since $\beta$ is only required to be positive for each menu. 
We study these general models in \cref{sec:menu_dep}.

We say that $X$ dominates $Y$ in the \emph{concave order}, denoted $X\ccv Y$ if $\E{g(X)}\ge \E{g(Y)}$ for all concave functions $g \colon \R \to \R$. We denote the strict part of $\ccv$ by $\sccv$. Recall that $X \ccv Y$ if and only if $Y$ is a mean-preserving spread of $X$. If $X \ccv Y$, then $U(X) \ge U(Y)$ for all vNM utilities $U$.\footnote{Since we defined vNM utilities as expectations of concave Bernoulli utilities, they are increasing in $\ccv$.} We will often use the variables $S$ and $R$ when referring to lotteries where $S\sccv R$, in order to highlight that $S$ is \emph{safer} while $R$ is \emph{riskier}. While we focus on lotteries that can be ranked in the concave order, analogous results hold throughout for lotteries comparable by the increasing concave order; see \cref{sec:notions}. Finally, we denote by $C^k_+$ the set of $k$-times continuously differentiable functions defined on an open interval of $\R$ with strictly positive $k$th derivatives.

We are interested in comparing the choices of individual $A$ with those of $B$. Each individual $i$ is associated with a utility function $U_i$ and stochastic choice model $\rho_i$. We will often refer to agents and their associated stochastic choice models interchangeably. We say $A$ is \emph{weakly more risk-averse than} $B$ if  $U_A(X) \geq U_A(y)$ implies $U_B(X) \geq U_B(y)$ for any $X \in \cL$ and degenerate lottery $y$.\footnote{This is \cite{yaari1969remarks}'s notion of comparative risk, which coincides with the Arrow-Pratt notion of comparative risk since $U_A$ and $U_B$ are vNM utility functions. I.e., $U_A$ being more risk-averse than $U_B$ is equivalent to the concavity of $u_A\circ {u_B}^{-1}$, when $u_A$ and $u_B$ are the corresponding Bernoulli utility functions to $U_A$ and $U_B$. Indeed, CARA and CRRA utilities are totally ordered by this order.} We say $A$ is \emph{more risk-averse than} $B$ if $A$ is weakly more risk-averse than $B$ and $B$ is not weakly more risk-averse than $A$.

\begin{definition}\label{def:paradoxical}
    When $A$ is more risk-averse than $B$, we say that there is a \emph{reversal} if there are lotteries $S \sccv R$, where
 \begin{align*}
     \rho_A(S,R) < \rho_B(S,R).
 \end{align*} Moreover, we say that $S$ and $R$ \emph{generate a reversal}.
\end{definition} 

This means that $A$, the more risk-averse individual, exhibits more risk-tolerant behavior by choosing the safer lottery $S$ less frequently than $B$. $A$'s behavior can also be seen as more risk-tolerant in the following sense: Let $s_A$ and $s_B$ denote $A$'s and $B$'s respective choice probabilities of $S$. The ex-ante lottery that $A$ receives is the compound lottery resulting in $S$ with probability $s_A$ and in $R$ with probability $1-s_A$, and likewise for $B$. Since $s_A < s_B$, $A$'s compound lottery is riskier than $B$'s, so $A$ evidently makes riskier choices than $B$. This interpretation naturally extends to the multinomial setting; see \cref{sec:multi}.

\medskip
Of course, we might expect that despite $A$ being more risk-averse, she still chooses $S$ less frequently than $B$ due to heterogeneous noise. For example, if $A$ experiences much more noise than $B$, we do not want to predict $A$ to choose safer options more frequently than $B$ for all comparisons. Indeed, we do not want to preclude this possibility; rather, we want to ensure that reversals are not a built-in assumption of the utility specifications themselves. 

When $A$ is more risk-averse than $B$, we say that their two utilities $U_A$ and $U_B$ \emph{imply reversals} in a given noise framework if, for every pair of admissible noise specifications for  $A$ and $B$, the resulting $\rho_A$ and $\rho_B$ yield reversals. Thus, if $U_A,U_B$ imply reversals, it is impossible to accommodate choice behavior in which the more risk-averse individual $A$ chooses the safer option more frequently than $B$ even with arbitrarily heterogeneous noise. 

For example, $U_A$ and $U_B$ imply reversals in the logit framework if for every pair of logit scale parameters $\beta_A$ and $\beta_B$, there exist $S \sccv R$ such that
\[\rho_A(S,R)=\frac{\ee^{\beta_A U_A(S)}}{\ee^{\beta_A U_A(S)}+\ee^{\beta_A U_A(R)}} < \frac{\ee^{\beta_B U_B(S)}}{\ee^{\beta_B U_B(S)}+\ee^{\beta_B U_B(R)}}=\rho_B(S,R).\]

\medskip
In the following sections, we show that Constant Absolute Risk Aversion (CARA) and Constant Relative Risk Aversion (CRRA)---the most commonly used parametric utility specifications in both theory and practice---imply reversals in many noise frameworks. CARA and CRRA utilities are defined by $\cara_a(X)=\E{\scara_a(X)}$ and $\crra_a(X)=\E{\scrra_a(X)}$, where $a> 0$ is the coefficient of absolute/relative risk aversion and
\begin{align*}
\scara_a(x)=
    \frac{1-\ee^{-ax}}{a}, \qquad \qquad 
\scrra_a(x)= \begin{cases}%\label{eq:CRRA}
    \frac{x^{1-a}-1}{1-a} & a\neq 1 \\
    \ln(x) & a=1.
\end{cases} 
\end{align*}
 The coefficients reflect varying levels of risk aversion, with higher coefficients corresponding to greater aversion to risk.\footnote{Note that there are many possible ways of parameterizing CARA and CRRA preferences, since applying a positive affine transformation to each Bernoulli utility does not change the underlying preference/coefficient. Nevertheless, we will show that our results hold for every possible parameterization of these families.} To ensure that $\crra_a(X)$ is well defined, we restrict our attention to $X$ that is non-negative and bounded away from $0$ whenever referring to CRRA.

Since we are interested in understanding when utility functions imply reversals in stochastic choice frameworks, 
we do not focus on models that make the prediction that every individual will choose the safer alternative deterministically. 
This includes random parameter models in which each individual's vNM utility is drawn randomly from a one-parameter family $\{U_\omega\}_{\omega\in \Omega}$ and $\rho(X,Y)=\mP_\omega(U_\omega(X)\ge U_\omega(Y))$, where $\mP_\omega \in \Delta(\Omega)$. See \cref{app:rpm_reversals} for further discussion about random parameter models.

\section{Fechnerian Models}\label{sec:main}
We first consider models of the form 
\begin{align}\label{eq:Fechner}
  \rho(X,Y)=F(U(X)-U(Y)),  
\end{align}
where $F \colon \R \to [0,1]$ is a $C^1_+$ function satisfying $F(t)+F(-t)=1$.
 We denote by $\cF$ the set of all such functions and refer to a model of the form $(U,F)$, where $U$ is a vNM utility and $F \in \cF$, as a \emph{Fechnerian expected utility} (FEU) model.\footnote{The assumptions that $F$ must be continuously differentiable with positive derivative can be relaxed without affecting our results, provided these assumptions hold at $0$. Note that we do not assume that $\lim_{t \to \infty}F(t)=1$, nor that $\lim_{t \to -\infty}F(t)=0$, although this will be the case for the examples we study. 
 } In these models, the larger the difference between $U(X)$ and $U(Y)$, the higher the probability that $X$ is chosen.

FEU models are among the most widely used stochastic choice models \cite*[see, e.g.,][]{becker1963stochastic,loomes2002microeconometric}. An important subclass of FEU models yields choice probabilities
\begin{align}\label{eq:iid}
     \mP(U(X)+\varepsilon \geq U(Y) + \varepsilon')=F(U(X)-U(Y)), 
\end{align}
where $\varepsilon$ and $\varepsilon'$ are i.i.d.\ continuous random variables and $F$ is the CDF of $\varepsilon-\varepsilon'$. This subclass includes many widely used discrete choice models such as the logit and probit models, where the random shocks are Gumbel and Gaussian, respectively. 

 In this section, we focus on FEU models and show that the CARA and CRRA specifications display a paradox: they predict that a more risk-averse individual will sometimes choose a riskier asset more frequently than a more risk-tolerant individual under any Fechnerian noise specifications.

\begin{proposition}\label{th:main}
Let $U_A$ and $U_B$ be distinct $\cara$ $(\crra)$ utilities. Then $U_A$ and $U_B$ imply reversals in the FEU framework. 
\end{proposition}

Since CARA (CRRA) preferences are totally ordered by risk-aversion, without loss of generality suppose $A$ is more risk-averse. 
\cref{th:main} says that for any $F_A,F_B \in \cF$, there nevertheless exist lotteries $S\sccv R$ such that $\rho_A(S,R)=F_A(U_A(S)-U_A(R))<F_B(U_B(S)-U_B(R))=\rho_B(S,R)$.

\cref{th:main} is an impossibility result for the commonly used CARA/CRRA-expected-utility Fechnerian models since reversals arise under arbitrary pairs of noise structures and risk coefficients. In the case of $F_A=F_B$, i.e., when individuals face the same noise structures, it was shown by \cite{wilcox2011stochastically} and \cite{apesteguia2018monotone} that the probability of choosing the safer option from a pair is non-monotone in the CARA/CRRA coefficients.

In practice, however, it is common to model heterogeneity in both preferences and noise and to estimate them jointly \citep{hey1994investigating,von2011heterogeneity}. Indeed, subsequent work has argued that non-monotonicity is particularly worrisome under homoskedastic noise and has advocated allowing heteroskedasticity for calibration of each individual's noise to their utility scale \citep{barseghyan2018estimating,keffert2024stochastic}.

Yet, perhaps surprisingly, \Cref{th:main} demonstrates that non-monotonicity persists even when noise structures vary arbitrarily across individuals. Thus, no matter how we jointly estimate risk coefficients and noise structures for two individuals, the resulting estimated models $(\cara_a,F_A)$  and $(\cara_b,F_B)$ will predict reversals. Because reversals are inherent to these models, we must be cautious when using these estimates to rank individuals by risk aversion.

Note that choice probabilities are not invariant to reparameterizations of CARA and CRRA utilities. For example, consider the model $(U,F)$, where $U$ is a CARA or CRRA utility and $F \in \cF$. Let $V(X)=\frac{1}{c}U(X)+d$ for some $c > 0$. While $V$ represents the same preference as $U$, replacing $U$ with $V$ has the same effect on choice probabilities as changing the noise structure from $F$ to $G$ which is defined by $G(t)=F(t/c)$. However, since $G$ is also a member of $\cF$, it follows from \Cref{th:main} that reparameterizations of the CARA or CRRA utility family cannot resolve the non-monotonicity. 

\subsection{Scale-Family Heteroskedasticity}
To develop intuition for \Cref{th:main}, we first consider individuals with distinct CARA or CRRA utilities $U_A$ and $U_B$ with respective coefficients $a$ and $b$, where $a>b$ and whose noise structures are related by a scale factor. That is, $F_A(t) = F(t/\sigma_A)$ and $F_B(t) = F(t/\sigma_B)$ for some  $\sigma_A,\sigma_B > 0$ and common $F \in \cF$.\footnote{This is the case, for example, when $F_A$ and $F_B$ are both normal or both logistic CDFs.} The choice probabilities of $S$ over $R$
are then given by \begin{align*}\label{eq:safe_a}
\rho_A(S,R)=F\left(\frac{U_A(S)-U_A(R)}{\sigma_A}\right), \qquad \rho_B(S,R)=F\left(\frac{U_B(S)-U_B(R)}{\sigma_B}\right).%\label{eq:safe_b}
\end{align*} Since $F$ is strictly increasing, the pair $(S,R)$ generates a reversal (i.e., $\rho_B(S,R) > \rho_A(S,R)$) if and only if
\begin{equation}\label{eq:constant_ratio}
    \frac{U_B(S)-U_B(R)}{U_A(S)-U_A(R)} > \frac{\sigma_B}{\sigma_A}.
\end{equation} 

Note that $\frac{\sigma_B}{\sigma_A}$, which is the relative noise level, does not depend on the lotteries. Given this, the next lemma immediately implies that to generate a reversal, we can start with any $S \sccv R$, and get a reversal by adding to both a large enough constant $x$ (in the CARA case) or multiplying both by a large enough constant (in the CRRA case). 

\begin{lemma}\label{lem:infty_ratio}
For any $S\sccv R$, $a,b>0$, and all $x \in \R$ and $k>0$,
\begin{enumerate}
\item $\displaystyle \frac{\cara_b(S+x)-\cara_b(R+x)}{\cara_a(S+x)-\cara_a(R+x)} = C_1 \cdot \ee^{(a-b)x}$ 
\item $\displaystyle \frac{\crra_b(k \cdot S)-\crra_b(k \cdot R)}{\crra_a(k \cdot S)-\crra_a(k \cdot R)} = C_2 \cdot k^{a-b},$ 
\end{enumerate} 
for some positive constants $C_1$ and $C_2$.
\end{lemma} 

\cref{lem:infty_ratio} is proved in \cref{ap:th:main}. Note that when $a > b$, the ratios are strictly increasing in $x$ and $k$ and tend to infinity. Moreover, these properties hold for any $S\sccv R$. We thus have the following result which strengthens the existence of reversal pairs established in \cref{th:main} when noise structures are related by a scale factor.
\begin{proposition}\label{rem:scale}
Let $\sigma_A,\sigma_B > 0$ and $F \in \cF$. Let $F_A(t)=F(\frac{t}{\sigma_A})$, $F_B(t)=F(\frac{t}{\sigma_B})$,  and $S\sccv R$. Let $U_A,U_B$ be distinct CARA utilities and let $V_A,V_B$ be distinct CRRA utilities. 
Then there exist unique $x_0$ and $k_0$ such that 
\begin{enumerate}
    \item $S+x$ and $R+x$ generate a reversal under the Fechnerian models $(U_A,F_A)$ and $(U_B,F_B)$
    if and only if $x > x_0$;
    \item $k S$ and $k R$ generate a reversal under the Fechnerian models $(V_A,F_A)$ and $(V_B,F_B)$
    if and only if $k > k_0$. 
\end{enumerate} 

\end{proposition}

Note that the critical values $x_0$ and $k_0$ given in \cref{rem:scale} depend on the noise scale parameters $\sigma_A,\sigma_B$, risk coefficients, and the baseline lotteries. The values can be derived from \eqref{eq:constant_ratio} and the positive constants $C_1$ and $C_2$ given in  \cref{lem:infty_ratio}. I.e., critical values $x_0$ and $k_0$ solve for $C_1\cdot\ee^{(a-b)x}=\frac{\sigma_B}{\sigma_A}$ and $C_2\cdot k^{(a-b)}=\frac{\sigma_B}{\sigma_A}$.\footnote{Exact expressions for $C_1$ and $C_2$ are given in \eqref{eq:critical}. We obtain
\[x_0=\frac{1}{a-b}\ln\left(\frac{\sigma_B}{\sigma_A}\frac{b}{a}\left(\dfrac{\E{\ee^{-aR}}-\E{\ee^{-aS}}}{\E{\ee^{-bR}}-\E{\ee^{-bS}}}\right)\right),\quad  k_0=\left(\frac{\sigma_B}{\sigma_A}\frac{1-b}{1-a}\left(\frac{\E{S^{1-a}}-\E{R^{1-a}}}{\E{S^{1-b}}-\E{R^{1-b}}}\right)\right)^{\frac{1}{a-b}}.\]} For example, when the CRRA coefficients $a$ and $b$ are $0.8$ and $0.3$, respectively, and the normal noise variances are $\sigma^2_A=0.5$ and $\sigma^2_B=1$, the critical values that generate reversals for lotteries given in \cref{tab:crra_reversal_scale} are all less than $3$, illustrating that reversals arise under empirically relevant risk coefficients and lotteries.  

\Cref{rem:scale} highlights two problems with the CARA and CRRA Fechnerian models. First, there is a disconnect between the standard Arrow-Pratt notion of comparative risk-aversion developed in the deterministic framework with the resulting stochastic choice predictions. Namely, more risk-averse individuals are not more likely to make risk-averse choices. This counterintuitive prediction is contradicted by empirical evidence that risk aversion is negatively correlated with decision noise, as measured by the frequency with which individuals choose the concave-order dominated option \citep{bruner2017does}. In terms of estimation, this means that even if Bob chooses safer lotteries more frequently than Anne, we may conclude that Bob is more risk-tolerant than Anne.\footnote{This holds under any consistent estimator.} This is the case, for example, if we only observe the choice probabilities of the scaled lotteries in \Cref{tab:crra_reversal_scale}.

Second, these models make the systematic prediction that those most inclined to choose safer options when stakes are small must become the least inclined once the stakes grow modestly larger. Even without a specific model of risk aversion in mind, this prediction is implausible.

\subsection{Arbitrary Fechnerian Noise Structures}
 
 So far we have considered the special case that $F_A$ and $F_B$ are related by a scale factor. To prove the more general case considered in \Cref{th:main}, we establish the following lemma, which shows that the properties of CARA and CRRA utilities given in \Cref{lem:infty_ratio} give rise to non-monotonicity for any $F_A,F_B \in \cF$.

\begin{lemma}\label{lem:ratio_implies}
    Let $U_A$ and $U_B$ be vNM utility functions with $U_A$ more risk-averse. Suppose that for each $M\in \R$ there exist lotteries $S\sccv R$ such that 
    \begin{align}\label{eq:ratio}
         \frac{U_B(S)-U_B(R)}{U_A(S)-U_A(R)}\ge M.
    \end{align}
  Then $U_A,U_B$ imply reversals in the FEU framework.
\end{lemma}
\Cref{lem:ratio_implies} identifies a sufficient condition on the vNM utility functions to yield the negativity result in \Cref{th:main}. By \Cref{lem:infty_ratio}, CARA and CRRA utilities satisfy this condition, and thus \Cref{th:main} follows immediately from these lemmas. 

Let $S \sccv R$ satisfying \eqref{eq:ratio} for some $M \in \R$. The key step of the proof of \Cref{lem:ratio_implies} is to construct $S'\sccv R'$ such that 
\[
\frac{U_B(S)-U_B(R)}{U_A(S)-U_A(R)}=\frac{U_B(S')-U_B(R')}{U_A(S')-U_A(R')},
\] while making the utility differences arbitrarily small. To this end, we consider the menu $(S,R_\lambda S)$, where $\lambda \in (0,1)$ and $R_\lambda S$ is distributed as a compound lottery that yields $R$ with probability $\lambda$ and $S$ with probability $1-\lambda$. As $\lambda\to 0$, the utility difference between $S$ and $R_\lambda S$ vanishes. \Cref{lem:ratio_implies} exploits the linearity of vNM utilities, which ensures that the ratio of utility differences is independent of $\lambda$. It also relies on the differentiability of $F_A$ and $F_B$ to approximate them around $0$ by affine functions.
\Cref{lem:ratio_implies} is proved in \Cref{ap:th:main}.

Note that by \Cref{lem:infty_ratio}, \eqref{eq:ratio} is satisfied for all concave-ordered lotteries under sufficient shifting/scaling. Moreover, for CARA, the utility difference $U(S+x)-U(R+x)$ tends to $0$ as $x$ tends to infinity (see \cref{lm:diff_zero} in \cref{sec:prop_cara}). For CRRA, on the other hand, $U(kS)-U(k R)$ increases with $k$ for relative risk coefficients less than unity. We therefore apply an additional transformation of scaling down the probability of receiving the risky prospect so that the utility difference is arbitrarily small. Based on these observations, the following proposition extends the non-monotonicity result of \Cref{rem:scale} to all $F_A,F_B \in \cF$.

\begin{proposition}\label{prop:stakes}
Let $F_A,F_B \in \cF$ and $S \sccv R$. Let $U_A$ and $U_B$ be distinct CARA utilities and $V_A$ and $V_B$ be distinct CRRA utilities. Then there exist $x_0$, $k_0 >0$ such that \begin{enumerate}
    \item $S+x$ and $R+x$ generate a reversal under the Fechnerian models $(U_A,F_A)$ and $(U_B,F_B)$ for $x > x_0$.
    \item $k\cdot S$ and $k\cdot (R_\lambda S)$ generate a reversal under the Fechnerian models $(V_A,F_A)$ and $(V_B,F_B)$ for $k > k_0$ and $\lambda$ small enough.
\end{enumerate}
\end{proposition}
\Cref{prop:stakes} is proved in \Cref{sec:prop_cara}. While CARA preferences are invariant to changes in background wealth, the above proposition shows that non-monotonicity occurs in FEU models under CARA utilities for any concave-ordered lotteries under sufficient background wealth. Likewise, while CRRA preferences are invariant to the scaling of stakes and to mixing with the safe lottery, these transformations can always generate a non-monotonicity.

\Cref{prop:stakes} demonstrates that non-monotonicity is a pervasive problem. It holds for all pairs of risk parameters, for all pairs of Fechnerian noise specifications, and for all concave-ordered lotteries after a transformation that preserves the concave ordering as well as the family of preferences being measured. 
\section{Weak Utility Models}\label{sec:menu_dep}
In the previous section, we studied models in which noise varies across individuals but not across menus for the same individual. 
In this section, we study a more flexible model that incorporates menu-dependent noise and show that the reversals under CARA/CRRA utilities persist.

Formally, $\rho$ has a  \emph{weak expected-utility} (WEU) representation $(U,F,\sigma)$ if
\begin{align}\label{eq:weak_util}
  \rho(X,Y) =F\left(\frac{U(X)-U(Y)}{\sigma(X,Y)}\right),  
\end{align} where $U$ is a vNM utility function, $F \in \cF$ is a Fechnerian noise structure, and $\sigma \colon \cL \times \cL \to (0,\infty)$ is a menu-dependent scale parameter satisfying $\sigma(X,Y)=\sigma(Y,X)$. Symmetry of $\sigma$ is required so that $\rho(X,Y)+\rho(Y,X)=1$.

An example of WEU is the \emph{menu-dependent probit model}. In this model, the shock terms $\varepsilon$ and $\varepsilon'$ in \eqref{eq:iid} follow normal distributions, but the identicality assumption is relaxed to allow their variances to depend on the lotteries. That is, 
\[\rho(X,Y)=\mP(U(X)+\varepsilon_X \ge U(Y)+\varepsilon_Y),\] where $\varepsilon_X$ and $\varepsilon_Y$ are independent and normal random variables with means zero and variances $\sigma_X^2$ and $\sigma_Y^2$.\footnote{In this case, $\rho$ can be written as \eqref{eq:weak_util} where $F$ is the CDF of the standard normal distribution and $\sigma(X,Y)=\sqrt{\sigma_X^2+\sigma_Y^2}$.} For example, $\sigma_X$ may depend on the variance of lottery $X$ in the menu. This can capture the intuitive idea that more variable lotteries may generate more noise, allowing for mistake probabilities to remain high even when stakes are high. 

Another example is the \emph{strict utility model} of \cite{luce1965preference} in which the choice probabilities are determined by the utility ratio
\begin{align*}
    \rho(X,Y)=\frac{U(X)}{U(X)+U(Y)}, 
\end{align*}
rather than by the utility difference as in FEU models (see \eqref{eq:Fechner}) of the previous section.
Indeed, any \emph{simply scalable expected utility} (SSEU) representation, which takes the form \begin{align}\label{eq:SSEU}
\rho(X,Y)=H(U(X),U(Y)),
\end{align} where $H:\R^2
 \to [0,1]$ satisfies $H(s,t)+H(t,s)=1$ and is strictly increasing in its first argument, has a WEU representation.\footnote{To see that the simply scalable model is a special case of the weak utility model, take 
\[
 \sigma(X,Y)=\begin{cases}
     \frac{U(X)-U(Y)}{F^{-1}(H(U(X),U(Y)))} & U(X) \neq U(Y)\\
     s > 0 & U(X)=U(Y)
 \end{cases} 
\] in \eqref{eq:weak_util} to derive 
$\rho(X,Y)=H(U(X),U(Y))$.}

As the examples above suggest, the menu-dependent scale parameter $\sigma$ allows for a highly flexible stochastic choice model. In fact, without imposing restrictions on $\sigma$, the sole behavioral implication of WEU is that individuals select their preferred lottery with probability greater than one-half.\footnote{Since $F(t/\sigma(X,Y)) \in \cF$ for all $(X,Y) \in \cL\times\cL$, if $U(X)-U(Y) \ge 0$, then $F\left(\frac{U(X)-U(Y)}{\sigma(X,Y)}\right)\ge \frac{1}{2}$ with equality if and only if $U(X)=U(Y)$. Conversely, if $U(X)>U(Y)$, then for any $p \in (\frac{1}{2},1)$, letting $\sigma(X,Y)=\frac{U(X)-U(Y)}{F^{-1}(p)}$ yields the desired probability.}  In order to rule out pathological noise assignments, we require that menus with similar lotteries are associated with similar menu-dependent scale parameters. 
We say that a sequence of lotteries $X_1,X_2,\dots$ converges to a lottery $X$ if for all %increasing and 
continuous 
functions $u$ it holds that \[
\E{|u(X_n)-u(X)|}\to 0.
\] In other words, $X_1,X_2,\dots$ converges to $X$ if the sequence of random utilities $u(X_1),u(X_2),\dots$ converges to $u(X)$ in $L^1$ for every %increasing and 
continuous Bernoulli 
$u$. We say that $\sigma$ is \emph{continuous} if whenever $X_n \to X$ and $Y_n \to Y$, it holds that $\sigma(X_n,Y_n) \to \sigma(X,Y)$. We say a WEU representation $(U,F,\sigma)$ is \emph{continuous} when $\sigma$ is continuous.\footnote{Formally, we require $\sigma$ to be sequentially continuous in the product topology on $\cL\times \cL$, when $\cL$ is equipped with the coarsest topology making each $X\mapsto u(X)$ continuous as an $L^1$-valued map. %This topology is the same whether we consider all continuous $u$ or only those that are increasing.
This topology is very fine (finer than the weak topology and any $L^p$ topology for $1\le p <\infty$), making many functions on $\cL$ continuous.}

Importantly, vNM utilities are continuous, i.e., $U(X_n) \to U(X)$ whenever $X_n \to X$, so we can allow $\sigma$ to depend continuously on $U$. Moreover, $\sigma$ is not required to depend only on the marginal distributions of the lotteries in the menu; it may, for example, depend on their joint distribution. 

\subsection{Characterization}
Our next theorem characterizes when two vNM utilities imply reversals in the continuous WEU framework. 
Given two increasing and concave Bernoulli utility functions $u_A$ and $u_B$, we say that $u_A$ is \emph{relatively more concave} than $u_B$, and write $u_A \rel u_B$ if there exists $k>0$ such that $k\cdot u_A(x)- u_B(x)$ is concave.\footnote{Recall that $u_A$ is more risk-averse than $u_B$ in the Arrow-Pratt sense if $u_A=f\circ u_B$ for some increasing and concave function $f$. This notion of comparative concavity is logically independent of $\rel$. In \cref{sec:notions}, we show that the stronger comparative risk notion of \cite{ross1981some} implies $\rel$.
} 
We extend this notion to vNM utilities $U_A=\E{u_A}$ and $U_B=\E{u_B}$, writing
$U_A \rel U_B$ if $u_A \rel u_B$. 
 For twice differentiable functions with $u_A''<0$, the condition that $u_A \rel u_B$ is equivalent to the boundedness of $u_B''(x)/u_A''(x)$.
 Note that the $\rel$ order is preserved under positive affine transformations of the two utility functions and is thus a relation on vNM preferences, i.e., it does not depend on utility representations. We write $U_A \norel U_B$ if it is not the case that $U_A\rel U_B$.

\begin{maintheorem}\label{th:finalboss}
Let $U_A$ and $U_B$ be vNM utilities with $U_A$ more risk-averse. Then $U_A,U_B$ imply reversals in the continuous WEU framework if and only if $U_A \norel U_B$. 
\end{maintheorem}

We emphasize that $U_A \rel U_B$ only implies the existence of noises for $A$ and $B$ that can accommodate no reversals. In the proof of sufficiency below, we show that many standard forms of noise, such as heteroskedastic Gaussian shocks, succeed at eliminating reversals. Of course, there will always exist other noises such that reversals will occur. For example, if $A$ experiences much more noise than $B$, then $A$ will make more mistakes and may choose some riskier options more frequently than $B$, even with $U_A \rel U_B$. On the other hand, when $U_A$ is not relatively more concave than $U_B$, reversals must occur regardless of $A$'s and $B$'s noises. Thus, $\rel$ characterizes when reversals are not driven by the utilities themselves, but can arise only from relative noisiness.

The necessity of $U_A\rel U_B$ to avoid reversals is proved in \Cref{ap:menu}. A key idea is to show that if $U_A$ is not relatively more concave than $U_B$, then the ratio \[\frac{U_B(S)-U_B(R)}{U_A(S)-U_A(R)}\] is unbounded over the set of lottery pairs $S \sccv R$. Under continuity of $\sigma_A$ and $\sigma_B$, we show that this unboundedness leads to reversals using an argument similar to that of \cref{th:main}. This argument relies on the fact that continuity imposes a bound on the relative noisiness of $A$ and $B$ for some lottery comparisons. We note that limitations on relative noisiness lead to reversals even under some discontinuous menu-dependent scale parameters. For example, for $(U_A,F,\sigma_A)$ and $(U_B,F,\sigma_B)$, if $\sigma_B(S,R)/\sigma_A(S,R)$ is bounded, there will necessarily be reversals even if $\sigma_A$ and $\sigma_B$ are not continuous.

We prove sufficiency of $U_A \rel U_B$ here by showing that even under a simple scale-family (e.g., a heteroskedastic probit model), without menu dependence, this condition suffices. Indeed, suppose that $ku_A - u_B$ is concave for some $k>0$. Let $F \in \cF$ (e.g., the CDF of the standard normal distribution) and let $F_A=F_B=F$ . Let $\sigma_A(X,Y)=1/k$ and $\sigma_B(X,Y)=1$ for all $X,Y$.
Then for any lotteries $S\ccv R$,
\begin{align}
    \rho_A(S,R)&= F_A(k(U_A(S)-U_A(R))) \label{eq:A1}\\ \rho_B(S,R)&=F_B(U_B(S)-U_B(R)).\label{eq:B2}
\end{align} Since $ku_A-u_B$ is concave, $kU_A(S)-U_B(S)\ge kU_A(R)-U_B(R)$. Hence, since $F_A=F_B$ is strictly increasing, \eqref{eq:A1} exceeds \eqref{eq:B2}, meaning that there are no reversals.\footnote{Note that for this direction, we can relax the requirement that $F \in \sm$.}

\medskip
\Cref{th:finalboss} establishes a simple test for identifying whether a pair of vNM utilities generates reversals. Indeed, for CARA and CRRA utilities, $\scara_b''(x)/\scara_a''(x)=(b/a)e^{-(b-a)x}$ and $\scrra_b''(x)/\scrra_a''(x)=(b/a)x^{\,a-b}$ are unbounded for $a> b$. 
Thus, it follows from \cref{th:finalboss} that reversals of Fechnerian CARA and CRRA models cannot be eliminated by allowing continuous menu-dependent scale parameters.

\begin{corollary}\label{th:menu_dep}
Let $U_A$ and $U_B$ be  distinct $\cara$ $(\crra)$ utilities. Then $U_A$ and $U_B$ imply reversals in the continuous WEU framework.
\end{corollary}

\Cref{th:menu_dep} highlights that the non-monotonicity result of the menu-independent FEU models (\Cref{th:main}) is not merely a byproduct of asymptotic properties of choice probabilities when stakes are increased. Indeed, in the following example scale parameters are engineered so that choice probabilities are scale-invariant, yet monotonicity does not obtain.

\begin{example}\label{ex:scale}
Let $a > b > 1$, $U_A=\crra_a$, and $U_B=\crra_b$. Let $F$ be the CDF of the standard normal distribution.  Let \begin{align*}
       \sigma_A(X,Y)={\E{X^{a-1}+Y^{a-1}}}^{-1},\qquad  \qquad  
       \sigma_B(X,Y)={\E{X^{b-1}+Y^{b-1}}}^{-1}.
    \end{align*}
Then for all $k > 0$ and for $i \in \{A,B\}$, we have
\begin{align*}
    F\left(\frac{U_i(k X)-U_i(kY)}{\sigma_i(kX,kY)}\right)&=F\left(\frac{U_i(X)-U_i(Y)}{\sigma_i(X,Y)}\right).
\end{align*}

\end{example}
Note that $\sigma_A,\sigma_B$ decrease as $X$ and $Y$ are scaled up, capturing the behavioral property that individuals pay more attention when facing higher stakes. Moreover, the standard deviations are chosen so that each individual's probability of choosing $\beta X$ over $\beta Y$ does not depend on $\beta>0$, as  $\sigma_A(\beta X,\beta Y)=\beta^{1-a}\sigma_A(X,Y)$ and \[U_A(\beta X)-U_A(\beta Y)=\beta^{1-a}(U_A(X)-U_A(Y)).\]

Since $\sigma_A(X_n,Y_n) \to \sigma_A(X,Y)$ whenever $X_n \to X$ and $Y_n \to Y$, the model is continuous. 
Thus, it follows from \Cref{th:menu_dep} that even these calibrated models cannot avoid reversals. 

Another example of continuous WEU is when $\rho$ admits an SSEU representation $(U,H)$ given in \eqref{eq:SSEU}, above, where $H$ is further assumed to be differentiable and its first partial derivative, $H_1$, is continuous and positive. Indeed, for \[
 \sigma(X,Y)=\begin{cases}
     \dfrac{U(X)-U(Y)}{F^{-1}(H(U(X),U(Y)))} & U(X) \neq U(Y)\\[1em]
     \dfrac{F'(F^{-1}(H(U(Y), U(Y))))}{H_1(U(Y), U(Y))} & U(X)=U(Y),
 \end{cases}
\] the WEU model $(U,F,\sigma)$ coincides with the SSEU model $(U,H)$. 
Moreover, $\sigma$ is continuous since $U$, $F^{-1}$, $F'$, $H$, and $H_1$ are continuous and the two cases coincide in the limit as $U(X_n)$ and $U(Y_n)$ tend to the same value.

While continuity is a standard assumption, it does rule out some popular models. For example, consider the constant error model $\rho$  where $\rho(X,Y)=p$ if $U(X)>U(Y)$ for some $p \in (\frac{1}{2},1)$ and $\rho(X,Y)=\frac{1}{2}$ if $U(X)=U(Y)$. This has a WEU model representation $(U,F,\sigma)$ where $F(1)=p$ and $\sigma(X,Y)=|U(X)-U(Y)|$ for $U(X)\neq U(Y)$. However, since $\sigma$ cannot be zero, there is no value $\sigma$ can take when $U(X)=U(Y)$ to make $\sigma$ continuous. 
Note, however, any two vNM utilities $U_A$ and $U_B$ with $U_A$ more risk-averse will not imply reversals in the constant error model framework as long as $A$ is associated with a higher $p$ since $\rho_A(S,R)=p_A$ and $\rho_B(S,R)=p_B$ for all $S \sccv R$. 
\medskip

One may guess, at this point, that reversals arise in WEU models under CARA and CRRA utilities because the Fechnerian noise and menu-dependent scale parameter are applied to the vNM utility $U$, which is merely a representation of the preference and carries no cardinal significance. A natural alternative would be to calculate, for each lottery $X$, its certainty equivalent $CE(X)=u^{-1}(\E{u(X)})$, and let \[
\rho(X,Y) =F\left(\frac{CE(X)-CE(Y)}{\sigma(X,Y)}\right).
\] However, the next result demonstrates that this adjustment does not resolve the issue. In fact, reversals persist even if we more generally consider utility representations of the form $f \circ U$ for an arbitrary $f \in C^1_+$.

\begin{corollary}\label{cor:transform}
Let $F_A,F_B\in \cF$, $f_A,f_B\in C^1_+$ and let $\sigma_A$ and $\sigma_B$ be continuous. For any distinct $\cara$ $(\crra)$ utilities $U_A$ and $U_B$, there are reversals for $(f_A\circ U_A,F_A,\sigma_A)$ and  $(f_B\circ U_B,F_B,\sigma_B)$.
\end{corollary}

\Cref{cor:transform} shows that reversals are a fundamental feature of CARA and CRRA preferences, not an artifact of the utility representation. The underlying mechanism is straightforward: for any transformation $f \in \sm$, the WEU model $(f\circ U,F,\sigma)$ has a WEU representation $(U,F,\tau)$, where the new menu-dependent scale parameter $\tau$ remains continuous.\footnote{Here, \[
 \tau(X,Y)=\begin{cases}
     \frac{U(X)-U(Y)}{f(U(X))-f(U(Y))}\sigma(X,Y) & U(X) \neq U(Y)\\[1em]
     \frac{\sigma(X,Y)}{f'(U(Y))} & U(X)=U(Y),
 \end{cases}
\] which is continuous since $U$, $f$, $f'$, and $\sigma$ are continuous and the two cases coincide in the limit.} Given this, the result follows directly from \cref{th:menu_dep}. Since inverses of $\scara$ and $\scrra$ belong to $\sm$, replacing CARA and CRRA expected utilities with their certainty equivalents cannot eliminate reversals.

\medskip
\Cref{th:menu_dep,cor:transform} demonstrate that CARA and CRRA utilities generate reversals whenever choices are noisy, even if we use very general utility representations and allow noise structures to vary across individuals and menus. Thus, if we hope to model the stochastic choice behavior of $A$ and $B$, where $A$ consistently chooses safer lotteries more frequently than $B$, we must forgo CARA and CRRA preferences. Moreover, the boundedness of $u_B''/u_A''$ is violated by many parametric families of utilities that are used to model risk aversion. In particular, we show  in \cref{ap:non_ex} that expo-power utilities---a popular generalization of CARA---violate the condition, implying reversals in the continuous WEU framework. Similarly, we demonstrate that \emph{hyperbolic absolute risk aversion} (HARA) utilities---a popular generalization of CRRA---generally imply reversals.

We take on these challenges in the next section, where we identify the HARA subfamilies that deliver consistent comparative statics %under standard models of noise, 
and then show how to construct simple vNM utility families that retain this property beyond HARA.

\section{Monotone Alternatives to CARA and CRRA}\label{sec:mono}

HARA utility is a two-parameter family with a property that the reciprocal of absolute risk aversion is affine in wealth, i.e., \[A_{\alpha,\beta}(x)=\frac{1}{\alpha+\beta x}.\] A standard parameterization of HARA utility is given by 
\begin{align*}
\shara_{\alpha,\beta}(x)=\begin{cases}
\ln(\alpha+x)& \beta=1 \\
    \frac{1}{\beta-1}(\alpha+\beta x)^{1-\frac{1}{\beta}} & \beta \neq 1,
\end{cases}
\end{align*} for $\alpha \ge 0$ and $ \beta,x > 0$. Note that $\shara_{0,\beta}=\scrra_{\frac{1}{\beta}}$. 

\begin{proposition}\label{prop:HARA}
    Let $u_A=\shara_{\alpha_A,\beta_A}$ and $u_B=\shara_{\alpha_B,\beta_B}$ and suppose that $u_A$ is more risk-averse than $u_B$. Then $u_A \rel u_B$ if and only if $\alpha_A < \alpha_B$ and $\beta_A=\beta_B$.
\end{proposition}

The proof of \cref{prop:HARA} is straightforward and is in \cref{ap:HARA}. 
 The coefficient $\beta$ is Wilson's \citeyearpar{wilson1968theory} cautiousness parameter, and these one-parameter subfamilies of HARA, therefore, are known as \emph{equicautious} \cite*[see, e.g.,][]{amershi83}. Equicautious HARA is central to the finance literature, where it was shown to characterize linear risk-sharing rules \citep{wilson1968theory,amershi83} and two-fund monetary separation \citep{CASS1970122}. \cref{prop:HARA} sheds new light on this classical family, demonstrating that equicautious HARA utilities can accommodate consistent comparative statics in stochastic choice unlike CARA, CRRA, expo-power utilities and HARA utilities with heterogeneous cautiousness parameters.

We conclude with an approach to constructing utility families ordered by both risk aversion and $\rel$.
Recall that for strictly increasing concave functions $f$ and $g$, $f$ is \emph{more concave} than $g$ if $f\circ g^{-1}$ is increasing and concave. The following proposition provides a simple way to construct a parametric utility family that can accommodate consistent comparative statics. 
\begin{proposition}\label{rem:linex}
 Let $f$ and $g$ be strictly increasing concave functions such that $f$ is more concave than $g$. Let $u_A(x)=af(x)+(1-a)g(x)$ and $u_B(x)=bf(x)+(1-b)g(x)$, where $1>a > b \ge 0$. Then $U_A=\E{u_A}$ is more risk-averse than $U_B=\E{u_B}$ and $U_A \rel U_B$. 
\end{proposition}
These utilities are parametrized by the coefficient on the more risk-averse Bernoulli utility function $f$, i.e., the higher the coefficient, the greater the risk aversion.\footnote{Indeed, \[-\frac{u_A''(x)}{u_A'(x)} \ge -\frac{u_B''(x)}{u_B'(x)} \iff (a-b)\left(-\frac{f''(x)}{f'(x)}+\frac{g''(x)}{g'(x)}\right) \ge 0.\]  Since $f$ is more concave than $g$, it follows that $a \ge b$ is equivalent to $U_A$ being more risk-averse than $U_B$.}
Moreover, $\frac{1-b}{1-a}u_A-u_B=\frac{a-b}{1-a}f$ is concave since $a > b$ and $f$ is concave, so $u_A \rel u_B$.\footnote{Note that even if $f\norel g$, for any $a\in (0,1), af(x)+(1-a)g(x) \rel g(x)$. This can be understood geometrically: the condition $f \rel g$ means that the line segment between $g$ and $f$ can be extended past $f$ while staying within the set of concave functions. In particular this means $f \norel g$ implies that $f$ lies on the boundary of the set of concave functions.} 

 Families of preferences that can be represented as mixtures of two utility functions were axiomatized by \cite*{kartik2024single} using a single-crossing condition. When $f$ and $g$ are CARA (CRRA) utilities with different coefficients, $u_A$ and $u_B$ correspond to sum-ex (sum-power) utilities \cite[see, e.g.,][]{farquhar1987constant,bell1988one,pedersen2001characterisation}. \cref{rem:linex} thus suggests a replacement of the CARA utility family or CRRA utility family for estimation. One specifies an upper and lower bound on absolute risk aversion, corresponding to two CARA utilities, and then estimates the weight on each utility. The same exercise can be done with bounds on relative risk aversion and CRRA utilities. In \cref{ap:sum}, we discuss properties of sum-ex and sum-power utility functions and how to interpret the parameters. 

\section{Multinomial Choice}\label{sec:multi}
In this section, we extend our results to multinomial choice. Given a finite menu $M\subset \cL$, we denote by $\rho^M\in \Delta(M)$ the choice probabilities for the menu $M$. 

In order to study reversals in the multinomial setting, we extend our notion for binary menus to larger menus. %There are a number of ways to extend the notions of a reversal to multinomial choice.
Recall that in  binary choice, lotteries $S\sccv R$ generate a reversal if a more risk-averse individual, $A$, assigns a lower choice probability to $S$ than a less risk-averse individual, $B$. Consequently, in terms of the ex-ante compound lotteries generated by their stochastic choices, $A$ is exposed to greater risk than $B$. We formalize this interpretation in the multinomial setting as follows. 

\begin{definition}\label{def:compound}
    When $A$ is more risk-averse than $B$, we say that there is a \emph{compound reversal}
if there is a menu $M$ such that for all concave functions $u\colon \R \to \R$
\[
\sum_{X\in M} \E{u(X)}\rho_B^M(X) \ge \sum_{X\in M} \E{u(X)}\rho_A^M(X),  
\] with strict inequality for at least one such $u$.
\end{definition}
When a compound reversal occurs at $M$, the compound lottery resulting from $A$'s stochastic choice is dominated by that of $B$ in the concave order. For example, consider the menu $M=\{X,Y,Z\}$ with the following distributions:
\[X=
\begin{cases}
    0 & \textrm{w.p. } \frac{1}{6}\\
    50 & \textrm{w.p. } \frac{4}{6}\\
    100 & \textrm{w.p. } \frac{1}{6}
\end{cases}, \quad 
Y=\begin{cases}
    0 & \textrm{w.p. } \frac{1}{3}\\
    75 & \textrm{w.p. } \frac{2}{3}
\end{cases}, \quad 
Z=\begin{cases}
    25 & \textrm{w.p. } \frac{2}{3}\\
    100 & \textrm{w.p. } \frac{1}{3}
\end{cases}.\]

Note that while $X,Y,Z$ share a mean of $50$, none of them are comparable in the
concave order. However, it is easy to verify that $X \sccv Y_{0.5}Z$. Thus, if $A$ is more risk-averse than $B$, then $\rho_A^M(Y)=\rho_A^M(Z)=\frac{1}{2}$ and
$\rho_B^M(X)=1$ gives a compound reversal. 

In fact, there are many ways to have compound reversals in this menu. Any $\rho_A^M$ and $\rho_B^M$ such that   $\rho_A^M(X)<\rho_B^M(X)$ and $\rho_A^M(Y)-\rho_A^M(Z)=\rho_B^M(Y)-\rho_B^M(Z)$ will lead to compound reversals. Indeed, such compound reversals arise under various parameters of multinomial logit models with CARA utilities, e.g., under absolute risk coefficients $0.357$ and $0.024$ for $A$ and $B$ and common logit scale parameter $1$.

 As in the case of binary choice, when $U_A$ is more risk-averse than $U_B$ we say that $U_A,U_B$ \emph{imply compound reversals} in a framework of noise, if for every pair of admissible noises there are compound reversals.

We study compound reversals in the \emph{additive perturbed utility} (APU) framework of \cite{fudenberg2015stochastic}, which is a generalization of Fechnerian choice to the multinomial setting. Choice probabilities in a menu $M$ are given by  
\begin{align}\label{eq:APU}
\rho^M = \argmax_{p\in \Delta(M)} \sum_{X\in M}U(X)p(X)-c(p(X)), 
\end{align}
 where $U$ is a vNM utility function and $c:[0,1]\to \R\cup\{\infty\}$ is strictly convex  and continuously differentiable on $(0,1)$ and satisfies $\lim_{p\to 0}c'(p)=-\infty$. Any function $c$ with these properties is referred to as a \emph{cost function}. Both $U$ and $c$ are fixed across menus. We say that $\rho$ has an \emph{APU representation} $(U,c)$ if $\rho$ coincides with \eqref{eq:APU}.  We say that an APU model $(U,c)$ is \emph{smooth} if $c\in C^2_+$. One prominent example of smooth APU is multinomial logit, which arises under entropic cost $c(p)=\frac{1}{\beta}p\ln (p)$ and yields
 \[\rho^M(X)=\frac{\exp(\beta U(X))}{\sum_{Y \in M}\exp(\beta U(Y))}.\]
 
 Our next result shows that reversals occur in smooth APU models with arbitrarily heterogeneous cost functions unless the $\rel$ order obtains.

\begin{maintheorem}\label{th:apu_rel}
Let $U_A$ and $U_B$ be vNM utilities with $U_A$ more risk-averse. Then $U_A,U_B$ imply compound reversals in the smooth APU framework if and only if $U_A \norel U_B$.

\end{maintheorem}

\cref{th:apu_rel} is proved in \cref{ap:APU}.

Given that APU is an extension of the FEU (see \cref{lem:apu}), we already know from \cref{th:finalboss} that reversals will necessarily occur even under heterogeneous costs if the $\rel$ order does not obtain.\footnote{As the multinomial logit example above demonstrates, compound reversals are not limited to binary choice when $\rel$ fails.} The message of this section is thus that when $\rel$ does obtain, reversals can be avoided even for multinomial choice. 

It is interesting to note that random parameter models do not preclude compound reversals in the multinomial setting; see \cref{app:cara_compound} for an example. We discuss this in depth in  \cite{sung2026induced}. Although compound reversals may be less of a concern than binary-choice reversals, they help clarify the extent to which random-parameter specifications preserve comparative statics in stochastic choice.

\section{Conclusion}
This paper establishes that the paradox long noted for CARA and CRRA utilities under homoskedastic Fechnerian noise is not an artifact of restrictive parametric assumptions, but instead reflects a deeper incompatibility between these utility forms and noisy choice. We show that these utilities imply reversals under highly flexible noise specifications---including heterogeneous noise across individuals and menu-dependent noise. 

We obtain a simple characterization of whether two utility functions can yield sensible predictions in the presence of noise. 
This condition rules out even the more general expo-power utilities and HARA utilities, and points to new parametric families that restore intuitive comparative statics. These families, which include equicautious HARA utility, sum-ex and sum-power utilities, provide empirically tractable, well-behaved alternatives for measuring risk preferences in noisy environments. Moreover, we show that these well-behaved alternatives yield monotone comparative statics even beyond binary choice.

We conclude by noting that the continuity requirement on noise assignments may be restrictive in some settings. For example, discontinuous noise assignments may capture interesting features of binary comparisons, such as salience or complexity. Understanding which discontinuous models generate reversals thus remains an interesting direction for future research.

\bibliography{refs}
\appendix
\section{Proof of \texorpdfstring{\Cref{th:main}}{Theorem 1}}\label{ap:th:main}
We start by proving \cref{lem:infty_ratio}, which shows that \[\frac{U_B(S)-U_B(R)}{U_A(S)-U_A(R)}\] is unbounded over the set of pairs $S\sccv R$. 
\begin{proof}[Proof of \Cref{lem:infty_ratio}]
Let $S \sccv R$ and $a, b>0$. For CARA, we have
\begin{align*}
    \frac{\cara_b(S+x)-\cara_b(R+x)}{\cara_a(S+x)-\cara_a(R+x)}
    %&=\frac{a}{b}\dfrac{\E{\ee^{-b(R+x)}}-\E{\ee^{-b(S+x)}}}{\E{\ee^{-a(R+x)}}-\E{\ee^{-a(S+x)}}}\\
    &=\frac{a}{b} \left(\dfrac{\E{\ee^{-bR}}-
\E{\ee^{-bS}}}{\E{\ee^{-aR}}-\E{\ee^{-aS}}}\right)\ee^{(a-b)x}.
\end{align*} 
% Next, we show the analogous result for CRRA utility when stakes are scaled up. 
For CRRA, when $a$ and $b$ are not equal to 1, we have
\begin{align*}
    \frac{\crra_b(k\cdot S)-\crra_b(k \cdot R)}{\crra_a(k\cdot S)-\crra_a(k \cdot R)}
    %&=\frac{1-a}{1-b}\left(\frac{\E{S^{1-b}}-\E{R^{1-b}}}{\E{S^{1-a}}-\E{R^{1-a}}}\right)\frac{k^{1-b}}{k^{1-a}}\\
    &= \frac{1-a}{1-b}\left(\frac{\E{S^{1-b}}-\E{R^{1-b}}}{\E{S^{1-a}}-\E{R^{1-a}}}\right)k^{a-b}.
\end{align*}
Since positive risk coefficients $a$ and $b$ for CARA/CRRA utilities correspond to risk-averse utilities, 
\begin{align}\label{eq:critical}
    C_1=\frac{a}{b} \left(\dfrac{\E{\ee^{-bR}}-
\E{\ee^{-bS}}}{\E{\ee^{-aR}}-\E{\ee^{-aS}}}\right), \quad C_2=\frac{1-a}{1-b}\left(\dfrac{\E{S^{1-b}}-\E{R^{1-b}}}{\E{S^{1-a}}-\E{R^{1-a}}}\right)
\end{align} are positive. For $a$ or $b$ equal to 1, the expression for $C_2$ is given by the corresponding limit as $a$ or $b$ tends to 1. 
\end{proof}

\cref{lem:infty_ratio} provides sequences of lottery pairs $(S_n,R_n)$ along which the ratio of expected-utility differences $(U_B(S_n)-U_B(R_n))/(U_A(S_n)-U_A(R_n))$ diverges, for CARA and CRRA utilities. \cref{lem:ratio_implies} establishes that when such a sequence exists there will be reversals. An important subtlety of \cref{lem:ratio_implies} is that the reversal-generating lotteries may not be part of this sequence. Indeed, \cref{lem:ratio_implies} relies on the following lemma, which establishes that under the additional condition that utility differences tend to zero, there will be lotteries in the sequence that generate reversals.

\begin{lemma}\label{lem:par_seq}
    Let $F_A,F_B\in \cF$ and let $U_A$ and $U_B$ be risk-averse vNM utility functions with $U_A$ more risk-averse. There exist $M\in \R$ and $\delta>0$ such that all lotteries $S\sccv R$ satisfying 
    \begin{align*}%\label{eq:ratio2}
         \frac{U_B(S)-U_B(R)}{U_A(S)-U_A(R)}\ge M \quad \text{and} \quad U_i(S)-U_i(R)<\delta
    \end{align*} for $i=A,B$ generate a reversal.
\end{lemma}
\begin{proof}[Proof of \cref{lem:par_seq}]
   Let $f_A(t)=\frac{d}{dt}F_A(t)$ and $f_B(t)=\frac{d}{dt}F_B(t)$. Since $F_A$ is continuously differentiable, $f_A$ is continuous at $0$ so that for each $\varepsilon_A>f_A(0)$ there is $\delta_A>0$ such that for each $t \in (0,\delta_A)$, $f_A(t) <\varepsilon_A$ and $F_A(t)<\frac{1}{2}+t\varepsilon_A $. Since $F_B$ is continuously differentiable and $f_B > 0$, for each $0 < \varepsilon_B < f_B(0)$, there exists $\delta_B>0$ such that for any $t\in (0,\delta_B)$, $f_B(t)> \varepsilon_B$ and $F_B(t) > \frac{1}{2}+t \varepsilon_B $. Let $\delta=\min\{\delta_A,\delta_B\}$ and let $M=\varepsilon_A/\varepsilon_B$. 

    Let $S$ and $R$ as in the statement of the lemma. Then \begin{align*}
        F_B(U_B(S)-U_B(R)) &> \frac{1}{2}+(U_B(S)-U_B(R)) \varepsilon_B \\
        &\ge \frac{1}{2}+(U_A(S)-U_A(R)) \varepsilon_A\\
        &> F_A(U_A(S)-U_A(R)).
    \end{align*}
\end{proof}

In light of \cref{lem:par_seq}, 
\cref{lem:ratio_implies} follows from the observation that for any lotteries $S\sccv R$ and $\varepsilon>0$, there exist lotteries $S'\sccv R'$ with the same utility difference ratio such that each utility difference is less than $\varepsilon$. 
\begin{proof}[Proof of \Cref{lem:ratio_implies}]
Let $F_A,F_B\in \cF$ and let $U_A$ and $U_B$ be risk-averse vNM utility functions with $U_A$ more risk-averse. By \cref{lem:par_seq}, there exist $M\in \R$ and $\delta>0$ such that all lotteries $S\sccv R$ satisfying \begin{align*}
         \frac{U_B(S)-U_B(R)}{U_A(S)-U_A(R)}\ge M \quad \text{and} \quad U_i(S)-U_i(R)<\delta
    \end{align*} for $i=A,B$ generate a reversal. By hypothesis, there exist lotteries $S\sccv R$ satisfying the first inequality.

    %The right hand side is positive, and the numerator and denominator of the left hand side are non-negative since $U$ and $V$ are concave. 

For $\lambda \in (0,1)$, let $R_{\lambda}S$ denote a lottery distributed as a compound lottery that yields $R$ with probability $\lambda$ and yields $S$ with probability $1-\lambda$.

Then $S \sccv R_{\lambda}S$ and 
\[ \frac{U_B(S)-U_B(R_{\lambda}S)}{U_A(S)-U_A(R_{\lambda}S)}=\frac{U_B(S)-U_B(R)}{U_A(S)-U_A(R)},\] by the linearity of expected utility, i.e., $U(R_{\lambda}S)=\lambda U(R)+(1-\lambda)U(S)$ for any vNM utility $U$. Moreover, as $\lambda$ tends to zero, $U_i(S)-U_i(R_{\lambda}S)$ for $i=A,B$ tend to zero. Hence, $S$ and $R_{\lambda}S$ generate a reversal for all $\lambda$ small enough.

\end{proof}

\section{Proof of \texorpdfstring{\Cref{prop:stakes}}{Proposition 2}}\label{sec:prop_cara}

\cref{prop:stakes} shows how to construct the reversal-generating lotteries of \cref{th:main} by generalizing \cref{rem:scale} to arbitrarily different noise structures. The CARA case relies on the following lemma about diminishing utility differences as background wealth increases.

\begin{lemma}\label{lm:diff_zero}
Let $S \ccv R$. Then 
\[\lim_{x \to \infty}\cara_a(S+x)-\cara_a(R+x)=0\] for all $a > 0$.
\end{lemma}
\begin{proof}
Since $S$ and $R$ have equal means, we suppose, without loss of generality, that $\E{S}=\E{R}=0$. Recall that $\scara_a(x)= \frac{1-\ee^{-ax}}{a}$ denotes the CARA Bernoulli utility under coefficient $a$, and let $m$ denote the essential infimum of $S$.\footnote{This is the largest value that $S$ exceeds with probability 1.} By monotonicity, $\scara_a(x)-\scara_a(m+x)\geq \scara_a(x)-\cara_a(S+x)$ and by concavity of $u_a$, since $\E{S}=0$, $\scara_a(x)-\cara_a(S+x) \geq 0.$  Moreover, since $m \le 0$ and $\scara_a$ is concave,
\begin{equation*}\label{eq:sand2}
\frac{d}{dx}\scara_a (m + x)\cdot|m|\geq \scara_a(x) - \scara_a(m + x).
\end{equation*}

Note that $\frac{d}{d x}\scara_a (m + x)=\ee^{-a(m +x)}$. Thus
\begin{equation*}
    \ee^{-a(m +x)}\cdot |m| \ge \scara_a(x)-\scara_a ( m + x) \ge \scara_a(x)- \cara_a(S+x)\ge 0.
\end{equation*}
Taking the limit as $x\to \infty$, we see that $\scara_a(x)-\cara_a ( S + x)\to 0$. 
 Since $S$ was an arbitrary mean-zero lottery and $\E{R}=0$, we have  \[\lim_{x \to \infty} \scara_a(x)-\cara_a(R+x)=0\]  as well, concluding the proof. 
\end{proof}
%Lemma~\ref{lm:diff_zero} applies to any family of concave and monotone utility functions whose marginal utilities converge to $0$ as $x$ goes to infinity.

%\Cref{prop:stakes} now follows from \cref{lem:infty_ratio,lem:par_seq,lm:diff_zero}.
\begin{proof}[Proof of \Cref{prop:stakes}]
    For the case of CARA utilities, \cref{lem:infty_ratio} establishes that for any pair of lotteries $S\sccv R$, the ratio
    \[
    \frac{U_B(S+x)-U_B(R+x)}{U_A(S+x)-U_A(R+x)}
    \] tends to infinity with $x$, while \cref{lm:diff_zero} establishes that the numerator and denominator go to zero. Hence, the result follows from \cref{lem:par_seq}. 

    For the CRRA case, \cref{lem:infty_ratio} establishes that for any pair of lotteries $S\sccv R$, the ratio
    \[
    \frac{U_B(k\cdot S)-U_B(k\cdot R)}{U_A(k\cdot S)-U_A(k\cdot R)}
    \] tends to infinity with $k$. As in the proof of \cref{lem:ratio_implies}, 
    \[
    \frac{U_B(k\cdot S)-U_B(k\cdot R_{\lambda}S)}{U_A(k\cdot S)-U_A(k\cdot R_{\lambda}S)}=\frac{U_B(k\cdot S)-U_B(k\cdot R)}{U_A(k\cdot S)-U_A(k\cdot R)},
    \] and as $\lambda$ tends to zero, $U_i(k\cdot S)-U_i(k\cdot R_{\lambda}S)$ tends to zero for $i=A,B$, and the result follows from \cref{lem:par_seq}. 
\end{proof}

\section{Proof of \texorpdfstring{\Cref{th:finalboss}}{Theorem 3}}\label{ap:menu}
\Cref{th:finalboss} provides a characterization of vNM utilities that lead to consistent predictions under some noise structures.  The proof of \Cref{th:finalboss} makes use of %\cref{prop:mon} along with 
the following lemmas.

Given two increasing and concave Bernoulli utility functions $u_A$ and $u_B$, we say that $u_A$ is \emph{absolutely more concave than} $u_B$, and write $u_A \abs u_B$ if $u_A-u_B$ is concave.\footnote{Unlike $\rel$, the relation $\abs$ is not invariant to positive affine transformations and thus does not define a relation on vNM preferences.} We extend this notion to vNM utilities $U_A=\E{u_A}$ and $U_B=\E{u_B}$, writing
$U_A \abs U_B$ if $u_A \abs u_B$. The following lemma states that this absolute comparative concavity notion characterizes when reversals do not occur under identical Fechnerian noise.

\begin{lemma}\label{prop:mon}
Let $F \in \cF$ and let $U_A$ more risk-averse than $U_B$. Then there are no reversals under the Fechnerian models $(U_A,F)$ and $(U_B,F)$ if and only if  $U_A\abs U_B$.
\end{lemma}
\begin{proof}[Proof of \Cref{prop:mon}]

Since Fechnerian noise structure $F \in \cF$ is strictly increasing, the choice probability of $S$ is higher for $A$ than $B$ if and only if the utility difference is higher for $A$ than $B$, i.e., 
\begin{align*}
    U_A(S)-U_A(R) \ge U_B(S)-U_B(R),
\end{align*} or equivalently, 
\begin{align}\label{eq:diff}
    U_A(S)-U_B(S) \ge U_A(R)-U_B(R).
\end{align} 
We now show that this inequality holds for all lotteries $S\ccv R$ if and only if $U_A \abs U_B$. 

Let $h=u_A-u_B$. If $h$ is concave, then \eqref{eq:diff} is satisfied by the definition of the concave order. Conversely, for the sake of contradiction, suppose that $h$ is not concave. Then there are $x,y\in \R$ and $\lambda\in (0,1)$ such that \[h(\lambda x+(1-\lambda) y)<\lambda h(x)+(1-\lambda)h(y).\] Let $S$ denote a degenerate lottery that always pays $\lambda x+(1-\lambda) y$, and let $R$ denote a lottery that pays $x$ with probability $\lambda$ and $y$ with complementary probability. Then 
\[
    U_A(S)-U_B(S) =h(\lambda x+(1-\lambda) y)<\lambda h(x)+(1-\lambda)h(y)= U_A(R)-U_B(R),
\] contradicting \eqref{eq:diff}. Thus $u_A-u_B$ is concave, meaning $U_A\abs U_B$.
\end{proof}

Recall that in a WEU model $(U,F,\sigma)$, the probability of choosing $X$ over $Y$ is given by  \[F\left(\frac{U(X)-U(Y)}{\sigma(X,Y)}\right).\] Let $f(t)=\frac{d}{dt}F(t)$. Note that $f$ is continuous since $F \in \cF$. In the following lemma, we establish compact convergence of $(\frac{1}{\sigma_n}f(\frac{t}{\sigma_n}))_n$ when $\sigma_n \to \sigma$. For notational convenience, we let 
\[
f_{\sigma}(t)=\frac{1}{\sigma}\,f\!\left(\frac{t}{\sigma}\right).
\]
\begin{lemma}\label{lem:uniform}
If $\sigma_n \to \sigma$, then $(f_{\sigma_n})_n$ converges compactly to $f_{\sigma}$.
\end{lemma}

\begin{proof}[Proof of \Cref{lem:uniform}]
Let $K\subset \R$ be compact. 
\begin{align*}
    \sup_{t \in K}\left|f_{\sigma_n}(t)-f_{\sigma}(t)\right|  &=\sup_{t \in K} \left|\frac{1}{\sigma_n}f\left(\frac{t}{\sigma_n}\right)-\frac{1}{\sigma}f\left(\frac{t}{\sigma}\right)\right|\\
    &=\sup_{t \in K} \left|\frac{1}{\sigma_n}f\left(\frac{t}{\sigma_n}\right)-\frac{1}{\sigma}f\left(\frac{t}{\sigma_n}\right)+\frac{1}{\sigma}f\left(\frac{t}{\sigma_n}\right)-\frac{1}{\sigma}f\left(\frac{t}{\sigma}\right)\right|\\
    &\le \left|\frac{1}{\sigma_n}-\frac{1}{\sigma}\right|\sup_{t \in K}\left|f\left(\frac{t}{\sigma_n}\right)\right|+\frac{1}{\sigma}\sup_{t \in K}\left|f\left(\frac{t}{\sigma_n}\right)-f\left(\frac{t}{\sigma}\right)\right|
\end{align*}
Note that since $\sigma_n \to \sigma$, there is a compact set $S$ such that $\frac{t}{\sigma_n}\in S$ for all $n$ high enough and $t\in K$. Since $f$ is bounded on $S$, the first term tends to $0$, and since $\frac{t}{\sigma_n}$ converges to $\frac{t}{\sigma}$ uniformly on $K$, and $f$ is uniformly continuous on $S$, the second term tends to $0$ as well.
\end{proof}

\begin{lemma}\label{lem:ratio_implies_MD}
    Let $U_A$ and $U_B$ be vNM utility functions with $U_A$ more risk-averse. Suppose that for each $M\in \R$ there exist lotteries $S\sccv R$ such that 
    \begin{align*}
         \frac{U_B(S)-U_B(R)}{U_A(S)-U_A(R)}\ge M.
    \end{align*}
  Then $U_A,U_B$ imply reversals in the continuous WEU framework.
\end{lemma}

\begin{proof}[Proof of \Cref{lem:ratio_implies_MD}]
     Let $U_A$ and $U_B$ be as in the statement of the lemma. Let $Z,Z',\Lambda \in \cL$ such that $Z$ and $Z'$ have the same distribution, and $\Lambda$ is  uniformly distributed on $[0,1]$ and independent of $Z$ and $Z'$. Let $F^A(t)=F_A\left(\frac{t}{\sigma_A(Z,Z')}\right)$ and $F^B(t)=F_B\left(\frac{t}{\sigma_B(Z,Z')}\right)$ and let $f^A(t)=\frac{d}{dt}F^A(t)$ and $f^B(t)=\frac{d}{dt}F^B(t)$. Let $\varepsilon_A > f^A(0)$ and $0<\varepsilon_B < f^B(0)$, which is possible since $f^B(0)$ is positive. Since $f^A$ and $f^B$ are continuous at $0$, there is $\delta>0$ such that for all $0 < t < \delta$, $f^A(t) < \varepsilon_A$ and $f^B(t) > \varepsilon_B$.
 
 By assumption, there are lotteries $X\sccv Y$ such that \[\frac{U_B(X)-U_B(Y)}{U_A(X)-U_A(Y)} \ge  \frac{\varepsilon_A}{\varepsilon_B}.\]
 
Since the inequality only depends on the distributions of $X$ and $Y$, we may choose $X$ and $Y$ to be independent of $\Lambda$. We define the random variables $S_\lambda$ and $R_\lambda$ by
\[
S_\lambda(\omega)=
\begin{cases}
    X(\omega) & \Lambda(\omega)\le \lambda \\
    Z(\omega) & \Lambda(\omega)> \lambda, 
\end{cases} \qquad
R_\lambda(\omega)=
\begin{cases}
    Y(\omega) & \Lambda(\omega)\le \lambda \\
    Z'(\omega) & \Lambda(\omega)> \lambda .
\end{cases}
\]

Note that for $\lambda \in (0,1)$,  $S_\lambda \sccv R_\lambda$ and 
    \[
    U_A(S_\lambda)-U_A(R_\lambda) = \lambda (U_A(X)-U_A(Y)), 
    \] since $U_A(Z)=U_A(Z')$. Hence, 
    
    \begin{align}\label{eq:mix_ratio}
        \frac{U_B(S_\lambda)-U_B(R_\lambda)}{U_A(S_\lambda)-U_A(R_\lambda)}=\frac{U_B(X)-U_B(Y)}{U_A(X)-U_A(Y)} \ge  \frac{\varepsilon_A} {\varepsilon_B},
    \end{align}
and $U_A(S_\lambda)-U_A(R_\lambda),U_B(S_\lambda)-U_B(R_\lambda) > 0$.

Let $F_\lambda^A(t)=F_A\left(\frac{t}{\sigma_A(S_\lambda,R_\lambda)}\right)$, $F_\lambda^B(t)=F_B\left(\frac{t}{\sigma_B(S_\lambda,R_\lambda)}\right)$ and let $f^A_\lambda(t)=\frac{d}{dt}F^A_{\lambda}(t)$ and $f^B_\lambda(t)=\frac{d}{dt}F^B_{\lambda}(t)$. Since $S_\lambda\to Z$ and $R_\lambda\to Z'$ as $\lambda\to 0$, continuity of $\sigma_A$ and $\sigma_B$, i.e., $\sigma_A(S_\lambda,R_\lambda)\to \sigma_A(Z,Z
')$ and $\sigma_B(S_\lambda,R_\lambda)\to \sigma_B(Z,Z')$, imply that $f^A_\lambda \to f^A$ and $f^B_\lambda \to f^B$ compactly (\cref{lem:uniform}). Thus, there is $\delta' > 0$ and $\lambda_A$ such that for all $\lambda\in (0,\lambda_A)$ and all $ t \in [0, \delta']$, $f^A_\lambda(t) < \varepsilon_A$ and $F^A_\lambda(t) < \frac{1}{2}+t\varepsilon_A$. Likewise, there is $\lambda_B$ such that for all $\lambda \in (0, \lambda_B)$ and all $t \in [0,\delta']$, $F^B_\lambda(t) > \frac{1}{2}+t\varepsilon_B$. 

For $\lambda$ small enough, \[U_A(S_\lambda)-U_A(R_\lambda) < \delta' \quad \text{ and } \quad     U_B(S_\lambda)-U_B(R_\lambda) < \delta'.\]
Thus, for $\lambda$ small enough, we have
%$\lambda = \min \{\lambda_f,\lambda_g\}$
\begin{align*}
F^B_\lambda(U_B(S_\lambda)-U_B(R_\lambda)) &> \frac{1}{2}+(U_B(S_\lambda)-U_B(R_\lambda)) \varepsilon_B \\
        &\ge  \frac{1}{2}+(U_A(S_\lambda)-U_A(R_\lambda)) \varepsilon_A\\
        &> F_\lambda^A(U_A(S_\lambda)-U_A(R_\lambda)).
    \end{align*}
The second inequality follows from \eqref{eq:mix_ratio}.

\end{proof}

\begin{proof}[Proof of \cref{th:finalboss}]
    Let $U_A$ be more risk-averse than $U_B$. In \cref{sec:menu_dep}, we showed that if $ U_A\rel U_B$ then there are continuous (indeed, constant) $\sigma_A$ and $\sigma_B$ and $F \in \cF$ such that there are no reversals under $(U_A,F,\sigma_A)$ and $(U_B,F,\sigma_B)$. We now show that when there will always be reversals if $U_A$ is not relatively more concave than $U_B$. Indeed, suppose that, for all $k>0$, it is not the case that $kU_A- U_B$ is concave, i.e., $kU_A$ is not absolutely more concave than $U_B$. Then, by \cref{prop:mon}, for each $k>0$, there are lotteries $S\sccv R$ such that 
    \[
    k(U_A(S)-U_A(R)) < U_B(S)-U_B(R).
    \] Hence, the result follows from \cref{lem:ratio_implies_MD}.
\end{proof}

\section{Other Comparative Risk Notions}\label{sec:notions}

Note that $U_A$ is more risk-averse than $U_B$ in the Arrow-Pratt sense if and only if, for all $x$, 
\begin{equation*}
\frac{u_B''(x)}{u_A''(x)} \le \frac{u_B'(x)}{u_A'(x)}.
\end{equation*}
While stochastic comparative risk depends on bounding the ratio of second derivatives by a constant, the traditional Arrow-Pratt notion of comparative risk
depends on bounding this ratio by the ratio of marginal utilities. Thus, when $u_B'(x)/u_A'(x)$ is bounded, the traditional notion of comparative risk is sufficient for the stochastic notion. On the other hand, \citeauthor{ross1981some}' \citeyearpar{ross1981some} stronger comparative notion requires \begin{equation}\label{eq:Ross}
\frac{u_B''(x)}{u_A''(x)} \le k \le \frac{u_B'(x)}{u_A'(x)},
\end{equation} for some $k > 0$ and all $x$, implying the boundedness of $u_B''/u_A''$, i.e., $u_A\rel u_B$. Moreover, \cite{ross1981some} shows \eqref{eq:Ross} is equivalent to the existence of $k$ such that $ku_A-u_B$ is concave and decreasing, while $u_A\rel u_B$ does not require $ku_A-u_B$ to be decreasing.

In fact, if $ku_A-u_B$ is concave and furthermore increasing, then $U_A=\E{u_A}$ and $U_B=\E{u_B}$ do not imply reversals under the Fechnerian or continuous WEU frameworks even for lotteries ordered by the more general increasing concave order.\footnote{We say that $X$ dominates $Y$ in the increasing concave order and write $X\ge_{\rm icv} Y $ if $E{u(X)}\ge E{u(Y)}$ for all increasing and concave functions $u$.} That is, we can accommodate 
\[
\rho_A(X,Y)\ge \rho_B(X,Y),
\] for all $X\ge_{\rm icv} Y$. Indeed, our proofs can be adapted to show that the existence of $k>0$ such that $ku_A-u_B$ is increasing and concave characterizes ``no reversals'' in the increasing concave order.

For example, the utility families constructed in \cref{rem:linex} satisfy this stronger condition and thus can avoid reversals for the larger increasing concave order. Moreover, it is easy to see that this mixture family can more generally avoid reversals in the even larger induced $\Omega$-order introduced by \cite{apesteguia2018monotone}.\footnote{Given a parameterized family of utility functions $\{U_\omega\}_{\omega \in \Omega}$ where $\Omega\subset \R$, we say that $X$ dominates $Y$ in the $\Omega$-order if there is an $\omega \in \Omega$ such that $U_\omega(X)\ge U_\omega(Y)$ and, moreover, for all $\omega'< \omega''$, $U_{\omega'}(X)\ge U_{\omega'}(Y)\implies U_{\omega''}(X)\ge U_{\omega''}(Y)$.} 
\section{Utilities Without Constant Risk Aversion}\label{ap:non_ex}

\subsection{Expo-power utilities}
We illustrate that the expo-power utility function, a two-parameter family proposed by \cite{saha1993expo} to capture increasing/decreasing absolute/relative risk aversion,\footnote{Notably, \cite{holt2002risk} use the expo-power utility function to model increasing relative risk aversion exhibited in their data.

}  suffers from the same problems as CARA and CRRA utilities. The increasing and concave expo-power utility function is given by 
\[u_{a,r}(x)=\frac{1-\ee^{-ax^{1-r}}}{a},\] for $x \ge 0$ and positive $a$ and $0\le r < 1$. Note that the absolute risk aversion is given by \[A_{a,r}(x)=\frac{r}{x}+a(1-r)x^{-r},\] for $x > 0$, which is decreasing for $r>0$ and constant for $r=0$.  Thus, unlike CARA or CRRA utility, the family of expo-power utilities, parameterized by $a$ and $r$, is not totally ordered by absolute/relative risk aversion. 

Importantly, for any distinct pairs of coefficients $(a_1,r_1)$ and $(a_2,r_2)$ such that $A_{a_1,r_1}(x) \ge A_{a_2,r_2}(x)$ for all $x$, so that $U_1=\E{u_{a_1,r_1}(X)}$ is more risk-averse than $U_2=\E{u_{a_2,r_2}(X)}$, it must hold that $r_1=r_2$ and $a_1 > a_2$. A simple calculation shows that when $r_1=r_2$ and $a_1 > a_2$, $u_{a_2,r_2}''(x)/u_{a_1,r_1}''(x)$ tends to infinity as $x \to \infty$.
Thus, we have the following corollary.
\begin{corollary}\label{cor:expo}
Let $U_A,U_B$ be expo-power utilities with $U_A$ more risk-averse. Then $U_A,U_B$ imply reversals in the continuous WEU framework.
\end{corollary}
\Cref{cor:expo} highlights that the paradoxical properties arising from CARA (CRRA) utilities are pervasive, and that they do not depend on the strong assumption of constant absolute (relative) risk aversion. 

\subsection{Hyperbolic Absolute Risk Aversion (HARA) Utilities}\label{ap:HARA}

\begin{proof}[Proof of \cref{prop:HARA}]
First, note that if $u_A$ is more risk-averse than $u_B$, then the reciprocal of risk aversion, given by $T_i(x)=\alpha_i+\beta_i x$ must be higher for $B$ than $A$ for all $x>0$, so $\alpha_A \le \alpha_B$ and $\beta_A \le \beta_B$, with at least one of these inequalities strict. 

Note that $u_i''(x)=-T_i(x)^{\frac{-1-\beta_i}{\beta_i}}$. In order to have $u_A \rel u_B$, 
\[\frac{u_B''(x)}{u_A''(x)}=\frac{T_A(x)^{1+\frac{1}{\beta_A}}}{T_B(x)^{1+\frac{1}{\beta_B}}}\] must be bounded which only holds when $\beta_A \ge \beta_B$. Thus, for $u_A$ to be more risk-averse than $u_B$ and $u_A \rel u_B$, it must be that $\beta_A=\beta_B$ and $\alpha_A < \alpha_B$.
\end{proof}

\section{Sum-ex and Sum-power Utility Functions}\label{ap:sum}
In \cref{sec:mono}, we demonstrated that for $h>\ell \ge 0$ the utility families 
\[
u_a(x;h,\ell)=a\cdot \scara_h(x)+(1-a)\cdot\scara_{\ell}(x) \qquad \qquad a\in (0,1)
\] and 
\[
v_a(x;h,\ell)=a\cdot \scrra_h(x)+(1-a)\cdot\scrra_{\ell}(x) \qquad \qquad a\in (0,1),
\] which are parametrized by the coefficient $a$ placed on the more risk-averse $\scara$ or $\scrra$ utility do not imply reversals in the FEU and continuous WEU frameworks. A utility function taking the form $u_a(x;h,\ell)$ is a sum-ex utility and one taking the form $v_a(x;h,\ell)$ is a sum-power utility. In this appendix, we briefly discuss the axioms that characterize the sum-ex and sum-power families as well as an interpretation of the risk parameter $a$. 

\cite{bell1988one,Bell1996MeasuringRiskReturnPortfolios}
showed that sum-ex and sum-power utilities are characterized by \emph{one-switch} properties on preferences over gambles and decreasing absolute risk aversion. 

\begin{definition}\label{def:switch}
    We say that a Bernoulli utility $u$ satisfies the \emph{additive one-switch rule} if, for all $X$ and $Y$ 
    \[
    \E{u(X+t)-u(Y+t)}
    \] is single-crossing in $t$. Likewise, we say that $u$ satisfies the \emph{multiplicative one-switch rule} if, for all $X$ and $Y$ 
    \[
    \E{u(kX)-u(kY)}
    \] is single-crossing in $k>0$.
\end{definition}
The family of CARA utilities is similarly characterized by a zero-switch condition, which requires that preferences over gambles are invariant to background wealth. However, we often observe individuals who become less risk-averse as their wealth levels increase, exhibiting \emph{decreasing absolute risk aversion} (DARA). Formally, we say that a Bernoulli utility function $u$ exhibits DARA if $-u''(x)/u'(x)$ is decreasing in $x$.

As argued by \cite{bell1988one}, the additive one-switch property is normatively appealing under DARA, since one may start to prefer a riskier gamble once their wealth exceeds some point and their aversion to risk diminishes. However, this preference should not be reverted at even higher levels of background wealth, at which the individual is even more risk tolerant. Similarly, \cite{Bell1996MeasuringRiskReturnPortfolios} argues that the multiplicative one-switch rule is a normatively appealing criterion for an individual comparing two investments with different stochastic returns. 
\cite{bell1988one,Bell1996MeasuringRiskReturnPortfolios} and \cite{pedersen2001characterisation} show that a concave infinitely differentiable Bernoulli utility function $u$ that exhibits DARA satisfies the additive (multiplicative) one-switch rule if and only if $u$ is a sum-ex (sum-power) utility function.

\medskip
We next provide an interpretation of the risk coefficient $a$ that is to be estimated for these families. Specifically, we show that the coefficient $a$ can be interpreted as a background wealth level. That is, the preference of an individual with a lower coefficient $a$ is identical to that of an individual with a larger coefficient and more background wealth. This equivalence is formalized in the next proposition.

\begin{proposition}\label{prop:coeff}
Let $1>a>b >0$ and $h>\ell>0$. Then there are $\alpha,\beta,t,k > 0$, $\gamma,\delta \in \R$ such that
\begin{align}\label{eq:coef}
u_b(x;h,\ell)&=\alpha\cdot u_a(x+t;h,\ell)+\gamma \qquad \text{for all } x\in \R\\
v_b(x;h,\ell)&=\beta \cdot v_a(k\cdot x;h,\ell)+\delta \qquad\,\,\, \text{for all } x>0. \label{eq:coef2}
\end{align}
\end{proposition}

\begin{proof}[Proof of \cref{prop:coeff}]
For sum-ex utility, let
\[
t=\frac{1}{h-\ell}\ln\!\Big(\frac{a(1-b)}{b(1-a)}\Big),\quad
\alpha =\frac{1-b}{1-a}\Big(\frac{a(1-b)}{b(1-a)}\Big)^{\frac{\ell}{h-\ell}},\quad
\gamma=\frac{b-\alpha \cdot a}{h}+\frac{(1-b)-\alpha(1-a)}{\ell}.
\]

For sum-power utility, if $\ell \neq 1$ and $h\neq 1$ let 
\[
k=\Big(\frac{a(1-b)}{b(1-a)}\Big)^{\frac{1}{h-\ell}}, \quad 
\beta=\frac{1-b}{1-a}\Big(\frac{a(1-b)}{b(1-a)}\Big)^{\frac{\ell -1}{h-\ell}},\quad 
\delta=\frac{a\cdot \beta-b}{1-h}+\frac{\beta(1-a)-(1-b)}{1-\ell}.
\]

Plugging in these parameters, we obtain \eqref{eq:coef} and \eqref{eq:coef2}. For $h$ or $\ell$ equal to $1$, taking the limits of the expressions for $\beta$ and $\delta$ works as well, completing the proof.
\end{proof}
\cref{prop:coeff} highlights a benefit that sum-ex and sum-power share with CARA and CRRA in the deterministic setting. One important feature of CARA utility is that background wealth does not affect preferences and thus does not need to be separately estimated. By \cref{prop:coeff}, the estimation of the coefficient on the higher CARA utility can itself be interpreted as an estimation of background wealth which does not, therefore, need to be separately estimated. Likewise, CRRA is often used as a preference specification over investments where it has the benefit that the level of stakes to be invested do not need to be estimated since they do not affect the preference. While sum-power preferences are not invariant to changes in stakes, the coefficient on the higher CRRA utility can be interpreted as revealing the level of stakes to be invested, which do not then need to be separately estimated.

\section{APU Models}\label{ap:APU}

To prove \cref{th:apu_rel}, 
we first show that the APU representation $(U,c)$ for $\rho$ reduces to the FEU model of \cref{sec:main} on binary menus under a smoothness assumption on the cost function. Recall that a cost function is a function $c \colon [0,1] \to \R \cup \{\infty\}$ that is strictly convex and continuously differentiable on $(0,1)$ and satisfies $\lim_{p \to 0}c'(p)=-\infty$. Also, recall that $\cF$ is the set of all $\sm$ functions $F \colon \R \to [0,1]$ satisfying $F(t)+F(-t)=1$.

\begin{lemma}\label{lem:apu}

Let $\rho$ have the smooth APU representation $(U,c)$. Then for all binary menus, \[\rho(X,Y)=F(U(X)-U(Y))\] for some $F \in \cF$. Moreover, for $p\in (0,1)$, $F^{-1}(p)=c'(p)-c'(1-p)$.

\end{lemma}
\begin{proof}[Proof of \cref{lem:apu}]

Fix an APU model $(U,c)$ and define $g(p)=c'(p) - c'(1-p)$. Since $c\in C_+^2$, $g\in C_+^1$. Moreover, since $g$ must satisfy $g(1-p)=-g(p)$ and $\lim_{p\to 0}c'(p)=-\infty$, we have  $\lim_{p \to 0}g(p)=-\infty$ and $\lim_{p \to 1}g(p)=\infty$. Let $F\colon \R \to (0,1)$ be the inverse of $g$. By the inverse function theorem, $F \in C_+^1$. Moreover, we have $1-F(t)=F(-t)$, so $F\in \cF$.

Finally, we show that the Fechnerian model $(U,F)$ and the APU model $(U,c)$ coincide for binary menus. Fix a binary menu $\{X, Y\}$ and let $p$ be the choice probability of $X$ under the APU model $(U,c)$. The maximization of the APU model yields the first-order condition:
\begin{align*}
    U(X) - U(Y) = c'(p) - c'(1-p)=g(p).
\end{align*} 
Applying $F$ to the first and last terms, we obtain $F(U(X)-U(Y))=p$.

\end{proof}

We now generalize \cref{prop:mon} to the broader APU framework.

\begin{lemma}\label{th:apu}
Let $c \in C_{+}^2$ and let $U_A$ more risk-averse than $U_B$. Then there are no compound reversals under the APU models $(U_A,c)$ and $(U_B,c)$ if and only if $U_A \abs U_B$.

\end{lemma}
 
\begin{proof}[Proof of \cref{th:apu}]
Let $\rho_A$ and $\rho_B$ have the APU representations $(U_A,c)$ and $(U_B,c)$, respectively. We first show that if $U_A \abs U_B$, there are no compound reversals. 

    Fix a menu $M$. Since $c$ is strictly convex, $\rho_A^M$ is the unique argmax of
    \[
    \sum_{X\in M}U_A(X)p(X)-c(p(X)),
    \] and $\rho_B^M$ is the unique argmax of
    \[
    \sum_{X\in M}U_B(X)p(X)-c(p(X)).
    \] If $\rho_A^M=\rho_B^M$, then there is no compound reversal, so suppose they are distinct.
    We thus have the following two inequalities:
    \begin{align*}
        \sum_{X\in M}U_A(X)\rho_A^M(X)-c(\rho_A^M(X)) &> \sum_{X\in M}U_A(X)\rho_B^M(X)-c(\rho_B^M(X)), \\
        \sum_{X\in M}U_B(X)\rho_B^M(X)-c(\rho_B^M(X)) &> \sum_{X\in M}U_B(X)\rho_A^M(X)-c(\rho_A^M(X)).
    \end{align*}

    Summing the inequalities gives
    \[
        \sum_{X\in M}(U_A-U_B)(X)\rho_A^M(X)> \sum_{X\in M}(U_A-U_B)(X)\rho_B^M(X).
    \]
Since $U_A \abs U_B$, $U_A-U_B=\E{v}$ for some concave function $v$, so we do not have a compound reversal.

Finally, by \cref{lem:apu}, $\rho_A$ and $\rho_B$ coincide with Fechnerian models $(U_A,F)$ and $(U_B,F)$ for binary choice. 
Since a reversal in binary choice is a compound reversal, by \cref{prop:mon}, the converse holds.

\end{proof}

\begin{proof}[Proof of \cref{th:apu_rel}]
First, we show $U_A \rel U_B$ implies that there are $c_A,c_B \in C_+^2$ such that there are no compound reversals under $(U_A,c_A)$ and $(U_B,c_B)$. Let $k > 0$ such that $kU_A \abs U_B$
and let $c \in C_+^2$. Let $\rho_A$ and $\rho_B$ have the APU representations $(kU_A,c)$ and $(U_B,c)$, respectively. By \cref{th:apu}, $\rho_A$ and $\rho_B$ do not yield compound reversals. Since $(U_A,\frac{c}{k})$ represents $\rho_A$ and $(U_B,c)$ represents $\rho_B$, setting $c_A=c/k$ and $c_B=c$ will ensure that there are no compound reversals.

Now suppose $U_A \norel U_B$. Then $\rho_A$ and $\rho_B$ coincide with Fechnerian models $(U_A,F_A)$ and $(U_B,F_B)$ for binary choice by \cref{lem:apu}, so it follows from \cref{th:finalboss} that $U_A$ and $U_B$ imply reversals in binary choice which are also compound reversals.

\end{proof}

\section{Random Parameter Models}\label{app:rpm_reversals}
As discussed in the main text, \cite{apesteguia2018monotone} raise the issue of reversals for CARA and CRRA models of binary choice under homogeneous Fechnerian noise. As an alternative, they advocate for \emph{random parameter} models, in which each individual's vNM utility is randomly drawn from a parameterized family. 

Formally, an \emph{ordered vNM family} is an indexed set of risk-averse vNM utilities $\cU = \{U_a\}_{a \in \cA}$ for some $\cA \subset \R$ and is ordered by risk aversion, i.e., for $a > b$, $U_a$ is more risk-averse than $U_b$. In a random parameter model on $\cU$, each individual $i$ is associated with a baseline coefficient $a_i$ in $\cA$ and additive random shock $\varepsilon_i$ that is non-negative and i.i.d.\ across individuals. We denote by $\xi_i=a_i+\varepsilon_i$ individual $i$'s random risk coefficient and by $U_{\xi_i}$ the corresponding random vNM utility. In this model, individual $i$ chooses lottery $X$ over $Y$ with probability\footnote{Technically, this is not a well-defined probabilistic choice rule when the probability of ties is non-zero. Rather than define a tie-breaking rule, we restrict ourselves to comparisons (and later on menus) in which ties occur with zero probability.}

$$
\rho_i(X,Y)=\mP(U_{\xi_i}(X)\ge U_{\xi_i}(Y)).
$$ 
We say that $A$ is more risk-averse than $B$ if $a_A > a_B$. 

\cite{apesteguia2018monotone} show that random parameter models do not lead to reversals in binary choice settings even with respect to a larger order of riskiness on lotteries that includes the concave order. Clearly, if $S \sccv R$, then every realization of $U_{\xi_i}$ would evaluate $S$ greater than $R$, since every possible realized utility function exhibits risk aversion. Since lotteries dominated in the concave order are never chosen in these models, there cannot, of course, be reversals. Indeed, \cite{apesteguia2018monotone} point out that random parameter models have limitations when applied to stochastic-dominance related lotteries. They suggest adding an additional trembling stage as a way to generate stochasticity for such pairs. 

We next show that the situation is more subtle in multinomial settings, where, for each menu $M$, the choice probabilities are given by  
\[
\rho_i^M(X)=\mP(U_{\xi_i}(X)=\max_{Y\in M} U_{\xi_i}(Y)).
\]

\subsection{Compound Reversals for CARA- and CRRA-based RPMs}\label{app:cara_compound}
We now construct examples of compound reversals in CARA- and CRRA-based random parameter models.
Let $p=10^{-5}$. Consider the menu $M=\{X,Y,Z\}$ with the following distributions:
\[X=
\begin{cases}
    20 & \textrm{w.p. } 1-p\\
    10 & \textrm{w.p. } p
\end{cases}, \quad 
Y=\begin{cases}
    23 & \textrm{w.p. } \frac{1-p}{2}\\
    20 & \textrm{w.p. } \frac{p}{2}\\
    19 & \textrm{w.p. } \frac{1-p}{2}\\
    10 & \textrm{w.p. } \frac{p}{2}\\
\end{cases}, \quad 
Z=\begin{cases}
    27 & \textrm{w.p. } \frac{1-p}{2}\\
    20 & \textrm{w.p. } p\\
    17 & \textrm{w.p. } \frac{1-p}{2}
\end{cases}.\]

For CARA random parameter models, let $\varepsilon$ be uniformly distributed on $[0,0.5]$. When $A$ and $B$ have absolute risk coefficients $a_A\approx 1.1571$ and $a_B=0.1$, under the random parameter model with $\varepsilon$ added to risk coefficients, we have 
 \[\rho^M_A(X)=\rho^M_A(Z)>0 \quad \text{and} \quad \rho^M_B(Y)=1.\] Since $Y \sccv X_\frac{1}{2}Z$, $Y \sccv W$ for any $W$ that is a mixture of $X,Y,Z$ with positive equal weights on $X$ and $Z$. Thus, we have a compound reversal. Similarly, we have a compound reversal of CRRA random parameter models from the same menu, when $A$ and $B$ have relative risk coefficients $a_A \approx 14.033$ and $a_B=2.2$, and $\varepsilon$ is uniformly distributed on $[0,10]$.

\end{document}